\documentclass[journal,twoside,web]{ieeecolor}
\usepackage{generic}
\usepackage{cite}
\usepackage{amsmath,amssymb,amsfonts}
\usepackage{graphicx}
\usepackage{hyperref}
\usepackage{algpseudocode}
\usepackage{textcomp}
\usepackage{xcolor}
\usepackage{stfloats}
\usepackage{booktabs}
\let\labelindent\relax
\usepackage{enumitem}
\usepackage{arydshln}
\usepackage{textcomp}
\usepackage{multirow}
\usepackage{mathtools}
\usepackage[thinc]{esdiff}
\usepackage{times}
\usepackage{bm}
\usepackage{algorithm}
\usepackage{algpseudocode}
\usepackage{accents}
\usepackage{caption}
\usepackage{subcaption}

\usepackage{tikz}
\usepackage[OT1]{fontenc} 

\hypersetup{hypertex=true,
	colorlinks=true,
	linkcolor=blue,
	urlcolor=black,
	anchorcolor=blue,
	citecolor=blue}

\newlength{\dhatheight}

\newtheorem{assumption}{Assumption}[section]
\newtheorem{definition}{Definition}[section]

\newtheorem{theorem}{Theorem}[section]

\newtheorem{lemma}{Lemma}[section]

\newtheorem{remark}{Remark}

\newcommand\norm[1]{\|#1\|}

\newcommand{\Var}{\mathbb{V}\mathrm{ar}}
\newcommand{\Cov}{\mathbb{C}\mathrm{ov}}
\newcommand{\Tr}[1]{\text{Tr}\left\{#1\right\}}

\def\BibTeX{{\rm B\kern-.05em{\sc i\kern-.025em b}\kern-.08em
		T\kern-.1667em\lower.7ex\hbox{E}\kern-.125emX}}
\begin{document}
\title{Fundamental Limitations of Data-Driven Control: \\ A Statistical Decision Perspective}
\author{Jiabao He, Feiran Zhao, Yushan Li, Yue Ju, Florian D\"orfler and H\r{a}kan Hjalmarsson
		\thanks{Jiabao He, Yushan Li, Yue Ju and H\r{a}kan Hjalmarsson are with the Department of Decision and Control Systems, School of Electrical Engineering and Computer Science, KTH Royal Institute of Technology, 100 44 Stockholm, Sweden.  (Emails: jiabaoh, yushanl, yuej, hjalmars@kth.se)}
\thanks{Feiran Zhao, and Florian D\"orfler are with the Automatic Control Laboratory, ETH Z\"urich, Physikstrasse 3, 8092 Z\"urich, Switzerland  (Emails:zhaofe@control.ee.ethz.ch, dorfler@ethz.ch)}
}
	
	\maketitle
	\vspace{-3mm}
\begin{abstract}
	Substantial research efforts have been devoted to the design of data-driven controllers; however, comparatively less is known about their statistical performance and fundamental limitations. This contribution develops a statistical decision framework for data-driven control, in which a controller is evaluated by its risk, defined as the expected performance degradation relative to the oracle model-based controller, and by its average risk over the parameter space. Within this framework, we propose a collection of design principles for data-driven controllers. We further derive lower bounds on risks by combining the bias-variance decomposition with the Cram\'er-Rao inequality. In particular, the optimal bias that attains the lower bound for the average risk is determined by calculus of variations, thereby making the bias-variance tradeoff in data-driven control explicit. Moreover, the derived bound reveals a ``waterbed'' effect in data-driven control: any improvement in risk relative to the lower bound over one region of the parameter space must be compensated by deterioration elsewhere. We illustrate the proposed framework on two canonical data-driven control problems: optimal feedforward control and the linear quadratic regulator benchmark. By comparing several representative data-driven controllers with the derived lower bounds, we sharpen the statistical interpretation of existing methods and reveal quantitative limitations that no controller design can avoid. 
\end{abstract}
	

\section{Introduction} \label{Sct1}
\vspace{-1mm}
Data-driven control, the task of designing controllers directly from process data, arises in a wide range of disciplines and has been studied for more than half a century. Early research paradigms include Iterative Feedback Tuning \cite{Hjalmarsson1998iterative}, correlation-based approaches \cite{Karimi2004iterative}, and Virtual Reference Feedback Tuning \cite{Campi2002virtual}; approaches that blend identification and control, such as dual control \cite{Wittenmark1995adaptive}, and identification for control \cite{Hjalmarsson2005experiment}; and many others. We refer to the survey \cite{Hou2013model} for an earlier overview of these developments. Recently, the area has witnessed a renewed surge of interest, driven by advances at the intersection of control, learning, optimization and related fields. Among recent developments,  two canonical benchmark problems have attracted considerable attention: data-driven model predictive control (MPC)  and linear quadratic regulator (LQR). In data-driven MPC, Willems' behavioral framework \cite{Willems2005note} has been leveraged in \cite{Coulson2019data,Berberich2020data,Liu2022data}, where collected input-output trajectories are incorporated into the constraints together with regularization terms in the control objective. The subspace predictive control \cite{Favoreel1999spc} has been further developed in \cite{Breschi2023data,Smith2024optimal,Dinkla2026closed}.  Recently, it has been shown that the behavioral and model-based approaches can be unified in a Bayesian framework \cite{Chiuso2025harnessing}. For data-driven LQR, representative contributions include reinforcement learning methods \cite{Fazel2018global,Recht2019tour,Yaghmaie2022linear}, robust design \cite{Dean2020sample}, certainty-equivalence (CE) method \cite{Mania2019certainty}, data informativity \cite{Van2020data}, direct data-driven methods  \cite{De2019formulas,De2021low,Dorfler2022bridging,Dorfler2023certainty,Zhao2025data,Zhao2025regularization}, and Bayesian method \cite{Schwaller2026bayesian}. We refer to the double special issues \cite{Dorfler2023dataA,Dorfler2023dataB} for many more contributions.

Despite the broad and rapidly growing literature on data-driven control, most existing works focus on the design and analysis of particular controllers. Comparisons between different data-driven control methods often still lead to the answer ``it depends'', as noted in \cite{Dorfler2023dataB}; consequently, determining which method should be preferred remains a nontrivial question. In theoretical analyses and numerical case studies, the advantage of one method over another is often demonstrated by showing that its performance is closer to that of the optimal model-based controller. This motivates us to raise a more foundational question: is there a limit to how close any data-driven controller can get to this ideal performance benchmark? Answering this question is important, because such a limit points directly to the inescapable hardness of the problem, thereby clarifying what is achievable, and conversely what is not achievable, in data-driven control.

This work addresses the above problem by formulating offline data-driven control as a statistical decision problem. A data-driven controller is viewed as a decision rule constructed from data, and its performance is assessed by a risk measure: the expected deterioration in performance relative to the oracle model-based controller. We further introduce the average risk, which aggregates this deterioration over the parameter space of interest. Within this framework, lower bounds on risks are derived. Since these bounds are independent of any particular design method, they serve as benchmarks for practical algorithms. Moreover, they reveal how the intrinsic difficulty of the problem is shaped by both the system under study and the information contained in data.

\vspace{-3mm}
\subsection{Related Work} \label{Sct1.1}

Although fundamental limitations are well established in classical model-based control \cite{Seron2012fundamental}, the limits on data-driven control are less studied.  Recent statistical analyses of data-driven control have studied stability and robustness guarantees \cite{Berberich2025overview}, suboptimality bounds \cite{Shi2025suboptimality}, and CE principle \cite{Liu2026certainty} for data-driven MPC; suboptimality bounds \cite{Dorfler2023certainty}, convergence \cite{Fazel2018global}, consistency \cite{Zeng2025noise}, and sample complexity for data-driven LQR \cite{Al2023sample,Dean2020sample,Mania2019certainty} and LQG \cite{Zheng2021sample}; relations between direct and indirect data-driven control methods \cite{Dorfler2022bridging,Fiedler2021relationship,Mattsson2024equivalence}, and what type of linear systems are hard to learn  \cite{Oymak2021revisiting,Jedra2022finite,Tsiamis2023statistical} and control \cite{Tsiamis2022learning}.

From the perspective of statistical decision theory, data-driven control closely resembles classical point estimation: in both cases, one infers quantities of interest from data and then makes decisions based on that information. This analogy has already influenced several developments in data-driven control. For instance, motivated in part by advances in regularized parameter estimation \cite{Pillonetto2022regularized}, model kernels that account for the expected deterioration in control performance have been proposed in \cite{Formentin2021control}; see also \cite{Care2023kernel} for a broader discussion of regularization in data-driven control. Another line of work addresses noisy data through an average-risk framework, leading to what may be termed Bayes control methods \cite{Scampicchio2019bayesian,Ferizbegovic2021bayes}. Meanwhile, fundamental limitations of point estimation are well established. A standard approach is to combine the bias-variance decomposition with the Cram\'er-Rao lower bound (CRLB) on the variance \cite{Fabian1977cramer,Bobrovsky1987some}, which leads to global Cram\'er--Rao bounds \cite{Bobrovsky1987some,Rao1992cramer} and related Bayesian bounds \cite{Gill1995applications,Brown1990information,Ben2009lower}. These developments suggest that tools from statistical decision theory \cite{Wald1950statistical,Lehmann1998theory} may provide a principled way to study not only the design of data-driven controllers, but also their fundamental limitations. 
A recent step in this direction was taken in \cite{Colin2024bias}, which proposed a bias-variance perspective on data-driven control. There, the bias and variance terms quantify the increase in control cost due to systematic error and controller variability, respectively. Moreover, lower bounds on the regret of online LQR were derived in \cite{Simchowitz2020naive,Cassel2020logarithmic,Ziemann2024regret,Bartos2026optimistic}, and local minimax lower bounds for offline LQR and LQG were given in \cite{Lee2023fundamental} and \cite{Lee2026fragility}, respectively. Compared with these works, this paper studies global lower bounds for data-driven control in a more general statistical decision framework. Moreover, we characterize the optimal bias that attains the lower bound for the average risk, thereby making the bias-variance tradeoff in data-driven control explicit.
\vspace{-3mm}	
\subsection{Contributions} \label{Sct1.2}
	
The main contributions of this paper are four-fold:

(1) We introduce a general analysis framework for data-driven control  from the perspective of statistical decision theory, in which a collection of decision rules for data-driven control is developed.

(2) By combining the bias-variance decomposition with the CRLB, we derive fundamental lower bounds for data-driven control in a dynamical system setting. In particular, we formulate the problem of lower bounding the average risk as a calculus of variations problem, which can be further cast as an ordinary differential equation (ODE) boundary value problem. The resulting solution yields a tight lower bound and characterizes the bias function when the bound is attained, thus making the bias-variance tradeoff in data-driven control explicit. 

(3) A ``waterbed'' effect: Our results suggest that, for a data-driven controller, any reduction in risk relative to the lower bound in one region of the parameter space must be compensated by increased risk elsewhere. We refer to this phenomenon as a ``waterbed'' effect\footnote{An analogous limitation is well known in control theory: Bode's sensitivity integral shows that reductions in sensitivity over certain frequency ranges must be compensated by increased sensitivity elsewhere \cite{Doyle2013feedback}.}.

(4) We apply the proposed framework to derive lower bounds for two canonical data-driven control problems: optimal feedforward control and the benchmark LQR. These bounds reveal quantitative limitations that no controller design can avoid and serve as benchmark for interpreting existing methods. Comparisons with several representative controllers show, in particular, that Bayesian methods are especially effective in approaching the lower bound on the average risk, in line with recent work \cite{Chiuso2025harnessing,Schwaller2026bayesian}.

\vspace{-3mm}
\subsection{Structure} \label{Sct1.4}

The remainder of this paper is structured as follows: Section~\ref{Sct2} introduces a statistical decision framework, in which a collection of decision rules is presented. Section~\ref{Sct4} describes a dynamical system setting for data-driven control. Section~\ref{Sct5} develops lower bounds for risks of data-driven control. In two parallel Sections~\ref{Sct6} and \ref{Sct7}, we apply the proposed framework to two data-driven control tasks: the optimal feedforward control problem and the LQR problem, respectively. Finally, this paper is concluded in Section~\ref{Sct8}. 

\vspace{-3mm}
\section{Preliminaries on Statistical Decision Theory} \label{Sct2}

This section introduces a statistical decision framework that provides a common viewpoint for point estimation and data-driven control using offline data.  Consider a setting in which the unknowns of the problem at hand are represented by a parameter vector $\theta \in \mathcal{D}_\theta \subset \mathbb{R}^{n_\theta}$, where $\mathcal{D}_\theta$ denotes the feasible set of $\theta$. Let the optimal decision be $h(\theta): \mathbb{R}^{n_\theta} \to \mathbb{R}^{n_h}$, and the loss for taking decision $\hat{h}$ be $L_\theta(\hat{h})$ when the true parameter is $\theta$. We focus on the quadratic loss function, which is widely used in the literature. Without loss of generality we assume that $n_h \leq n_\theta$, $L_\theta\left(h(\theta)\right) = 0$, and write the loss as
\begin{equation} \label{loss}
	L_\theta(\hat{h}) = \big(\hat{h} - h(\theta)\big)^\top W(\theta) \big(\hat{h} - h(\theta)\big),
\end{equation}
where $W(\theta) \succ 0$ is a problem-dependent weighting matrix, quantifying the importance of making a correct decision when the unknown is $\theta$. We only have access to $\theta$ indirectly through the realization of a random vector $Z$ with distribution $p(z;\theta)$, which constitutes the data. The decision rule will thus be a function of $Z$, $\hat{h} = \hat{h}(Z)$, and therefore a random variable. Within this framework, point estimation corresponds to the special case $h(\theta)=\theta$. In data-driven control, $h(\theta)$ denotes the optimal model-based controller, whereas $\hat h(Z)$ denotes a data-driven controller.

The long-term performance of the decision rule $\hat{h}$ is given by the risk (the expected loss)
\begin{equation} \label{risk}
	R_\theta(\hat{h}) := \mathbb{E}_\theta\left[L_\theta(\hat{h})\right] = \int L_\theta(\hat{h}(z))p(z;\theta)dz.
\end{equation}
Moreover, the performance of the decision rule $\hat{h}$ over all possible parameters in $\mathcal{D}_\theta$ is captured by the weighted average risk
\begin{equation} \label{average_risk}
r_\pi(\hat h)=\int_{\mathcal{D}_\theta} R_\theta(\hat h) \pi(\theta) d\theta,
\end{equation}
where $\pi(\theta)$ is a nonnegative weighting function. In Bayesian decision theory, $\pi$ is interpreted as a prior that encodes our subjective belief about the unknown system parameter. Then, minimizing \eqref{average_risk} leads to a Bayes decision rule; see \cite[Sect.~4.1]{Lehmann2006theory} and the references therein for further discussion. In this paper, however, we do not rely on this Bayesian interpretation. Instead, we use \eqref{average_risk} as a risk-aggregation criterion for constructing and evaluating decision rules with favorable performance over the parameter domain. 

Throughout the paper, we take $\pi \equiv 1$, so that all parameters in $\mathcal{D}_\theta$ are weighted equally. The corresponding average risk is denoted by $r(\hat h)$. It is straightforward to generalize our results to the weighted average risk $r_\pi(\hat h)$ in \eqref{average_risk} with a general weighting $\pi(\theta)$.

Since $L_\theta(\hat{h})$ is quadratic, we have
\begin{equation} \label{risk-decomposition}
	\begin{split}
		R_\theta(\hat{h}) &= \mathbb{E}_\theta\left[\left(\hat{h} - h(\theta)\right)^\top W(\theta) \left(\hat{h} - h(\theta)\right)\right]\\
		&=\Tr{W(\theta)\mathbb{E}_\theta\left[\left(\hat{h} - h(\theta)\right)\left(\hat{h} - h(\theta)\right)^\top\right]} \\
		&=\Tr{W(\theta)\left(b_\theta(\hat{h})b_\theta(\hat{h})^\top + \Var_{\theta}[\hat{h}]\right)} \\
		&=\norm{b_\theta(\hat{h})}_{W(\theta)}^2 + \Tr{W(\theta)\Var_{\theta}[\hat{h}]}, 
	\end{split}
\end{equation} 
where 
\begin{subequations} \label{risk-bias}
	\begin{align}
		b_\theta(\hat{h}) &= \mathbb{E}_\theta[\hat{h}] - h(\theta), \\
		\Var_{\theta}\left[\hat{h}\right] &= \mathbb{E}_\theta\left[\left(\hat{h} - \mathbb{E}_\theta[\hat{h}]\right) \left(\hat{h} - \mathbb{E}_\theta[\hat{h}]\right)^\top\right],
	\end{align}
\end{subequations}
are the bias and variance of $\hat{h}$, respectively. The third equality in \eqref{risk-decomposition} follows from the bias-variance decomposition. In \eqref{risk-decomposition}, the notation $\norm{b_\theta(\hat{h})}_{W(\theta)}^2 = {b_\theta(\hat{h})^\top W(\theta)b_\theta(\hat{h})}$, and $\Tr{W(\theta)}$ denotes the trace of $W(\theta)$, respectively.

The question now is how to design a decision rule $\hat h(Z)$ such that the risk $R_\theta(\hat{h})$ is minimized. It is immediate to see that there is no uniformly (in $\theta$) best decision rules. Indeed, setting $\hat h = h(\bar\theta)$ attains zero risk at $\theta = \bar\theta$. However, the proposed decision rule may be very bad when $\theta \neq \bar\theta$. Any data-driven rule $\hat h(Z)$ has nonzero variance and therefore cannot achieve zero risk. The question arises how to construct decision rules with low risk regardless of $\theta$. 
\vspace{-3mm}
\subsection{Design Principles} \label{Sct2.1}
In this subsection we outline some basic principles for constructing data-driven decision rules. We let $\mathcal{D}_h$ be the set of decision rules that we can choose from. Central to all techniques is that they aim at making the risk small but that they handle the fact that $\theta$ is unknown in different ways.

\subsubsection{Unbiased Decision Rules} \label{Sct2.1.1}

An unbiased decision rule $\hat{h}_{\text{UB}}$ is defined by 
\begin{equation} \label{eq:unbiased_ddc}
	\mathbb{E}_\theta[\hat{h}_{\text{UB}}(Z)] = h(\theta), \ \forall \ \theta\in \mathcal{D}_\theta.
\end{equation}
For certain families of distributions it is possible to derive an optimal unbiased decision rule, called the uniformly minimum risk unbiased decision rule, within this class.

\subsubsection{Maximum Likelihood (ML) Decision Rules} \label{Sct2.2.2}

The well-known ML estimate of $\theta$ is given by
\begin{equation}
	\hat{\theta}_{\text{ML}}(Z):=\underset{\theta\in \mathcal{D}_\theta}{\arg\max} \ p(z;\theta).
\end{equation}
By defining a likelihood function for the optimal decision $h(\theta)$, one may generalize from the ML estimate to a decision rule. We follow the approach in \cite{Zehna1966invariance} for such a generalization, where a likelihood function for a given value $\eta\in\mathcal{D}_h$ of $h(\theta)$ is defined by
\begin{equation}
p_\text{M}(z;\eta)
=
\sup_{\theta\in\mathcal{D}_\theta: \ h(\theta)=\eta} p(z;\theta).
\end{equation}
We call $p_\text{M}$ the max-likelihood and define the maximum max-likelihood decision rule for $h(\theta)$ as
\begin{equation}
\hat{h}_{\text{MML}}(Z)
=
\underset{\eta\in \mathcal{D}_h}{\arg\max}\ p_\text{M}(Z;\eta).
\end{equation}
Since $\hat{\theta}_{\text{ML}}(Z)$ is the ML estimator of $\theta$, for any $\eta\in\mathcal{D}_h$ we have
$p_\text{M}(Z;\eta)\leq p(Z;\hat{\theta}_{\text{ML}}(Z))$. On the other hand, for $\eta=h(\hat{\theta}_{\text{ML}}(Z))$, the parameter $\hat{\theta}_{\text{ML}}(Z)$ is feasible in the supremum defining $p_\text{M}$, and hence $p_\text{M}(Z;h(\hat{\theta}_{\text{ML}}(Z)))
=p(Z;\hat{\theta}_{\text{ML}}(Z))$. Therefore, if $h(\hat{\theta}_{\text{ML}}(Z))\in\mathcal{D}_h$ and the maximizer of $p_\text{M}(Z;\eta)$ is unique, it follows that
\begin{equation}\label{eq:MML}
\hat{h}_{\text{MML}}(Z)
= h(\hat{\theta}_{\text{ML}}(Z)).
\end{equation}

\subsubsection{Loss and Certainty-Equivalence (CE) Tuning}\label{Sct2.2.3}

Another approach is to estimate the loss function $L_\theta$, which is unknown due to its dependency on $\theta$, and choose the parameter that minimizes the estimated loss. The most straightforward way to estimate the loss is to replace the unknown $\theta$ in the loss $L_\theta$ by an estimate $\hat{\theta}(Z)$ and then in the given family $\mathcal{D}_h$ of decisions pick
\begin{equation}
	\hat{h}_{\text{CE}(\hat{\theta})}(Z)=\underset{\hat{h}\in \mathcal{D}_h}{\arg\min}\ L_{\hat{\theta}(Z)}(\hat{h}).
\end{equation}
When the decision rule $h(\hat{\theta}(Z))$ is included in $\mathcal{D}_h$, this will result in the decision
\begin{equation} \label{eq:CEP}
	\hat{h}_{\text{CE}(\hat{\theta})}(Z)=h(\hat{\theta}(Z)).
\end{equation}
This approach is commonly known as the CE principle since it uses the parameter estimate as if it is the true value in the optimal decision rule. Comparing \eqref{eq:CEP} with \eqref{eq:MML} we see that using the ML estimate of $\theta$ results in the ML estimate of $h(\theta)$, i.e., $\hat{h}_{\text{CE}(\hat{\theta}_{\text{ML}})}(Z)=h(\hat{\theta}_{\text{ML}}(Z))$.

\subsubsection{(Weighted) Average Risk Tuning} \label{Sct2.2.4}

A way to bypass the problem that the loss and risk functions that we would like to minimize are functions of $\theta$ is to instead replace the loss by its average in the $\theta$-domain. This leads naturally to the average risk $r(\hat h)$ defined in \eqref{average_risk}. Interchanging the order of integration gives
\begin{equation}
	r(\hat h)=\int\left(\int_{\mathcal{D}_\theta} L_\theta(\hat h)p(z;\theta)d\theta\right)dz.
\end{equation}
Hence, for each fixed observation $z$, we are led to minimize
\begin{equation}
	V_z(\hat h) = \int_{\mathcal{D}_\theta} L_\theta(\hat h) p(z;\theta)d\theta.
\end{equation}
Let $\hat h_{A(P)}(z)$ denote the minimizer of $V_z(\hat h)$, where the subscript $A$ stands for the average risk minimization and $P$ has the origin in the Pitman estimator \cite{Pitman1939estimation}. Since $\hat h_{A(P)}(z)$ minimizes the integrand pointwise in $z$, the resulting rule $\hat h_{A(P)}$ minimizes the average risk $r(\hat h)$. For the quadratic loss $L_\theta(\hat h)$, the unconstrained minimizer takes the form
\begin{equation} \label{eq:Bayes_control}
	\hat h_{A(P)}(z)=\overline W^{-1}(z) \overline H(z),
\end{equation}
where
\begin{subequations}
	\begin{align}
		\overline W(z)&:=\int_{\mathcal{D}_\theta} W(\theta)p(z;\theta)d\theta,\\
		\overline H(z)&:=\int_{\mathcal{D}_\theta} W(\theta)h(\theta)p(z;\theta)d\theta.
	\end{align}
\end{subequations}
It should be mentioned that to compute $\hat h_{A(P)}(z)$, the integrals $\overline W(z)$ and $\overline H(z)$ have to be computed, which in general is intractable and requires sampling methods \cite{Lehmann1998theory,Berger2013statistical}.

By the same argument as the average risk $r(\hat h)$ where $\pi\equiv 1$, the minimizer of the general $r_\pi(\hat h)$ in \eqref{average_risk}, denoted by $\hat h_{A(\pi)}(z)$, can be derived in a similar form to \eqref{eq:Bayes_control}.

\subsubsection{Composite Methods}\label{Sct2.2.5}

In statistical decision theory, additional design principles such as minimax and risk tuning methods may also be considered. Together with the four families of methods introduced above, these can be combined to form a broad class of composite methods. The basic idea is to use one method to construct an intermediate parameter estimate or decision rule, and then use this as input to another method. Altogether, this shows that the different method families are not isolated alternatives, but can be combined in many ways to construct decision rules.

\vspace{-3mm}
    \subsection{An Illustrative Example: Point Estimation} \label{Sct2.2}
    The main objective of this work is to derive lower bounds on the pointwise risk $R_\theta(\hat h)$ and the average risk $r(\hat h)$ for data-driven control problems in a dynamical system setting. Before turning to this control problem, we first revisit related results in a classical point estimation problem. This simpler setting provides a useful preview of the results  that will later arise, in an analogous form, in data-driven control.
    
Consider the scalar Gaussian model $Z \sim \mathcal{N}(\theta, \sigma_z^2)$ with $|\theta| \leq L$, for some constant $L < \infty$. We are interested in estimating $\theta$ from a realization $z$ of the random variable $Z$. Three estimators are considered: the ML estimator $\hat{\theta}_{\text{ML}} = z$, the trivial estimator $\hat{\theta}_0 = 0$ which assumes a particular value of $\theta$, and the Bayes estimator $\hat{\theta}_{\text{Bayes}}$ using the uniform prior over the interval $\mathcal{D}_\theta$, given by $\hat{\theta}_{\text{Bayes}}(z)
	=\frac{\int_{-L}^{L}\theta\exp(-\frac{(z-\theta)^2}{2\sigma_z^2})d\theta
	}{\int_{-L}^{L}\exp(-\frac{(z-\theta)^2}{2\sigma_z^2})d\theta}$. 

The loss of an estimator $\hat{\theta}$ is taken as $L_\theta(\hat{\theta}) = \big(\hat{\theta} - \theta\big)^2$. Then, the corresponding risk $R_\theta(\hat h)$ and the average risk $r(\hat h)$ can be defined as in \eqref{risk} and \eqref{average_risk}, respectively. Using Theorems~\ref{Thm:CRLB} and~\ref{Theorem-EL-Equation-vector} in this work (to be presented later in Section~\ref{Sct5}), we obtain that for any estimator, the average risk $r(\hat{\theta})$ over the interval $\mathcal{D}_\theta = [-L,L]$ is lower bounded by
\begin{equation} \label{Gaussian-uniform-prior-lower-bound}
	\begin{split}
		r(\hat{\theta}) 
		&\geq \int_{-L}^{L}  
		\underbrace{(b_\theta^*)^2 
		+ \bigl((b_\theta^*)'+1\bigr)^2 \sigma_z^2}_{f\bigl(\theta,\sigma_z\bigr)}\,d\theta \\
		&= 2L\sigma_z^2
		\left(1-\frac{\tanh(L/\sigma_z)}{L/\sigma_z}\right),
	\end{split}
\end{equation}
where $b_\theta^* = -\sigma_z\frac{\sinh(\frac{\theta}{\sigma_z})}{\cosh(\frac{\theta}{\sigma_z})}$ is the optimal bias that attains the lower bound, and $(b_\theta^*)'$ is its derivative w.r.t. $\theta$. The point estimation bound \eqref{Gaussian-uniform-prior-lower-bound} appeared in \cite{Young1971error,Ben2009lower} and is recovered here as a special case of Theorems~\ref{Thm:CRLB} and~\ref{Theorem-EL-Equation-vector}.

\begin{figure}
	\centering
	\includegraphics[scale=0.35]{\detokenize{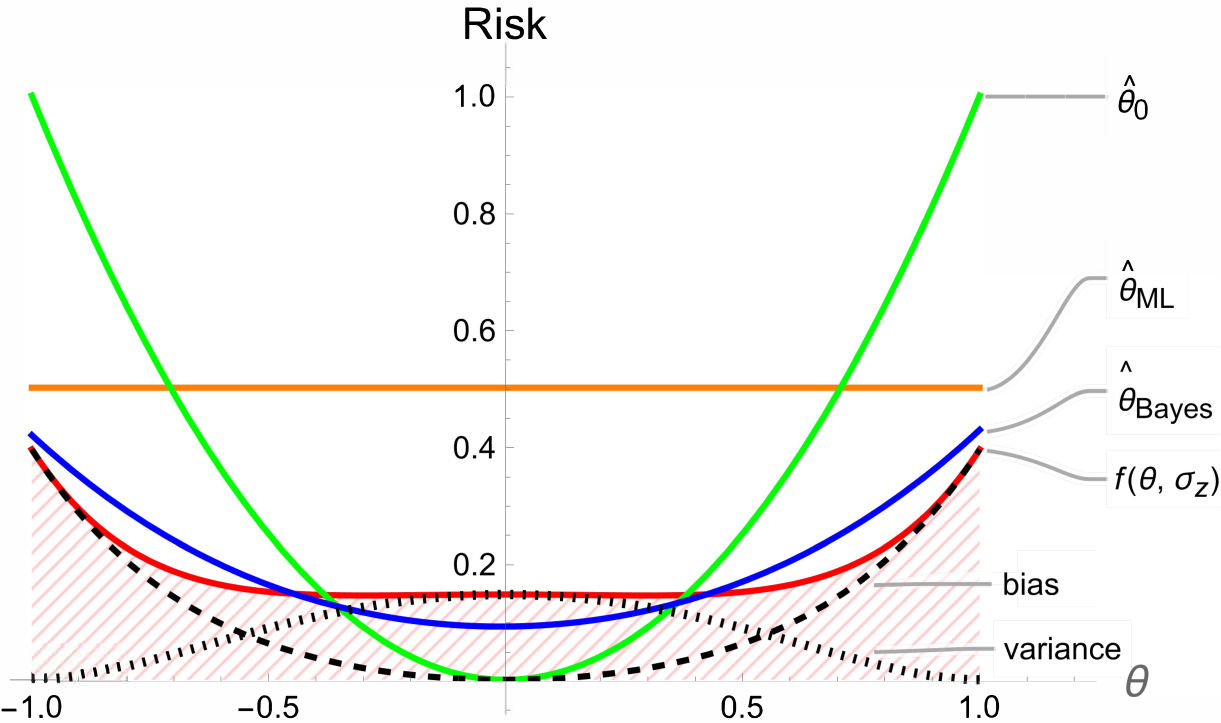}}
	\caption{Risks $R_\theta(\hat{\theta})$ of the three estimators over the interval $[-1,1]$: the average risks of $\hat{\theta}_{\text{ML}}$, $\hat{\theta}_{0}$, and $\hat{\theta}_{\text{Bayes}}$ are 1.0, 0.667, and 0.392, respectively. The lower bound, shown by the hatched red region, is 0.372. Moreover, the optimal squared bias and variance associated with $f\bigl(\theta,\sigma_z\bigr)$, the integrand in \eqref{Gaussian-uniform-prior-lower-bound}, are given.}
	\label{F1:MSE}
\end{figure}

Figure~\ref{F1:MSE} shows the pointwise risk $R_\theta(\hat{\theta})$ of three estimators above on $\mathcal{D}_\theta$, together with $f\bigl(\theta,\sigma_z\bigr)$, for $\sigma_z^2=0.5$ and $L = 1$. It can be observed that $\hat{\theta}_{\text{ML}}$ has a constant risk, and $\hat{\theta}_{\text{Bayes}}$ has the smallest average risk among the three estimators. Moreover, $\hat{\theta}_0$ is competitive only near $\theta=0$ and performs poorly elsewhere. It should be mentioned that \eqref{Gaussian-uniform-prior-lower-bound} lower bounds the area below any curve in Figure~\ref{F1:MSE} rather than having a pointwise lower bound on $R_\theta(\hat{\theta})$. In particular, an estimator may have risk below $f\bigl(\theta,\sigma_z\bigr)$ on part of the interval without causing any contradiction. However, its average risk over $\mathcal{D}_\theta$ cannot be smaller than the red hatched region in Figure~\ref{F1:MSE}, as illustrated by $\hat{\theta}_{0}$ and $\hat{\theta}_{\text{Bayes}}$. Consequently, reducing the pointwise risk on some subintervals must be compensated by increased risk elsewhere. We refer to this phenomenon as a ``waterbed'' effect, in analogy with the fundamental limitation of sensitivity shaping in classical control. Moreover, the two terms in $f\bigl(\theta,\sigma_z\bigr)$, i.e., the squared optimal bias $(b_\theta^*)^2$ and the term $\bigl((b_\theta^*)'+1\bigr)^2 \sigma_z^2$ which lower bounds the variance, are shown in Figure~\ref{F1:MSE}. Around the center of the parameter space, the squared bias is small; however, this comes at the price of a relatively large variance. In contrast, near the boundary of the parameter space, the variance decreases substantially, but only by allowing the squared bias to increase. Thus, the decrease of one term is accompanied by the increase of the other. This rise-fall behavior is exactly the bias-variance tradeoff.

    \vspace{-3mm}
    
\section{A Dynamical System Setting} \label{Sct4}

We are now prepared to consider a dynamical system setting. We denote the input and output of a dynamical system by $u(t)$ and $y(t)$, respectively. The observed output/input signal sequence that is available for data-driven control will be denoted by $Z^{N}=\{y(t),u(t)\}_{t=1}^N$. In a general form, the data is assumed to be generated as
\begin{subequations} \label{eq:general_dynamical_sysm}
	\begin{align}
		g_0(Z^0) &= 0, \\
        y(t) &= g_t(Z^{t-1})+e(t), \quad t \in \{1,2,\ldots, N\},
	\end{align}
\end{subequations}
where $g_t(\cdot)$ is a deterministic function, $\{e(t)\}_{t=1}^N$ is a sequence of independent random vectors with pdfs $p_t(e)$ and has a zero mean, and $Z^{t-1}$ represents all past output/input data up to the time $t-1$. The functions $g_t$ and the pdfs $p_t$ belong to a parametrized family
\begin{equation} \label{eq:Bayes_pdfs_family}
    \left\{\{g_t(\cdot;\theta),p_t(e;\theta)\}_{t=1}^{\infty}:\theta\in \mathcal{D}_\theta\right\},
\end{equation}
where $\theta$ is an unknown parameter, and $\mathcal{D}_\theta$ is an open set, i.e., there exists a $\theta\in \mathcal{D}_\theta$ such that $g_t(Z^{t-1})=g_t(Z^{t-1};\theta)$ and $p_t(e)=p_t(e;\theta)$. 

Using the chain rule of probability, together with the independence of the innovations, we obtain
\begin{equation} \label{eq:Bayes_pdfs}
p(Z^N;\theta)=\prod_{t=1}^{N} p_t(\varepsilon_t(Z^t;\theta);\theta),
\end{equation}
where $\varepsilon_t(Z^t;\theta):=y(t)-g_t(Z^{t-1};\theta)$ denotes the one-step prediction error induced by $\theta$. Using \eqref{eq:Bayes_pdfs} we can define the log-likelihood function and the score function, as
\begin{align}
	L_N(Z^N;\theta) &:= \log p(Z^N;\theta) = \sum_{t=1}^{N}\ell_t(Z^t;\theta), \\
    S_N(Z^N;\theta) &:= \frac{\partial}{\partial \theta} \log p(Z^N;\theta) = \sum_{t=1}^{N}s_t(Z^t;\theta),
\end{align}
where
\begin{equation*}
	\begin{split}
		\ell_t(Z^t;\theta) &:=\log p_t(\varepsilon_t(Z^t;\theta);\theta), \\
        s_t(Z^t;\theta) &:= \left.
\frac{\frac{\partial}{\partial \theta}p_t(e;\theta)+\frac{\partial}{\partial e}p_t(e;\theta)\Psi_t(Z^{t-1};\theta)}
{p_t(e;\theta)}
\right|_{e=\varepsilon_t(Z^t;\theta),}
	\end{split}
\end{equation*}
with $\Psi_t(Z^{t-1};\theta):=\frac{\partial}{\partial \theta}\varepsilon_t(Z^t;\theta)$.

To specify the conditions on $\{g_t\}$ and $\{p_t\}$, we will use the exponential forgetting framework in Ljung's seminal work \cite{Ljung1978convergence}. We will follow the version in \cite{Hjalmarsson1993aspects} which employs more flexible conditions than those in \cite{Ljung1978convergence}.

\begin{definition} [Exponentially Forgetting]  \label{Def:EF}
Let $\{z(t)\}_{t=1}^{\infty}$ be a stochastic process. A doubly indexed stochastic process $\{z_s^\circ(t), t\in T\subseteq\mathbb{N}\}$ is said to be a finite memory approximant of order $\gamma$ to $\{z(t)\}$, if $z_s^\circ(t)$ is independent of $\{z(k), k\le s\}$ and of $\{z_l^\circ(k), l\le k\le s\}$, $z_t^\circ(t)=0$, and there exists $C<\infty$ and $\tilde{\lambda}<1$, independent of $s$ and $t$, such that
\begin{equation}
\mathbb{E}\left[|z(t)-z_s^\circ(t)|^\gamma\right]^{1/\gamma}\le C\tilde{\lambda}^{t-s}.
\end{equation}
A stochastic process is said to be exponentially forgetting (EF) of order $\gamma$ if it has a finite memory approximant of order $\gamma$.
\end{definition}

According to the definition above, the remote past of an EF stochastic process is forgotten at an exponential rate. This is colloquially known as fading memory, a reasonable condition for most systems. 

Throughout this work, we impose the following assumption.
\begin{assumption} \label{Assmp_Fisher}
	The Fisher information matrix
	\begin{equation} \label{eq:FIM}
		I_{F,N}(\theta)
		:= \mathbb{E}_\theta\left[
		S_N(Z^N;\theta)S_N(Z^N;\theta)^\top
		\right]
	\end{equation}
	is well defined, i.e., $I_{F,N}(\theta)\succ 0$.
\end{assumption}

This assumption is standard in the literature \cite{Brown1990information,Lee2023fundamental}. In particular, it is satisfied when $\{y(t)\}_{t=1}^{\infty}$ in \eqref{eq:general_dynamical_sysm} is EF of order $q\cdot\gamma$, with $q>1$ and $\gamma>2$, and the family $\{p_t(\cdot;\theta):\theta\in \mathcal{D}_\theta,\, t\in T=\mathbb{N}_N\}$ satisfies some regularity conditions (see \cite{Brown1990information,Ben2009lower}).

Having specified the dynamical system setting and the data-generating mechanism, we now are ready to study the performance limits of data-driven control.

\vspace{-3mm}
\section{Risk Lower Bounds in Data-Driven Control} \label{Sct5}

In the point estimation problem reviewed in Section~\ref{Sct2}, the lower bound for the average risk $r(\hat{\theta})$ is established using the celebrated CRLB. In this section, we extend this idea to data-driven control, where the target is not the parameter $\theta$ itself, but a general function $h(\theta)$, representing the optimal model-based controller.

\vspace{-3mm}
\subsection{Lower Bound for the Pointwise Risk} \label{Sct5.1}

We have the following lower bound on the pointwise risk $R_\theta(\hat h)$ for data-driven control.
\begin{theorem} \label{Thm:CRLB}
	Let $\hat{h}(Z^N)$ satisfy $\mathbb{E}_\theta[\norm{\hat{h}(Z^N)}^2]<\infty$ for all $\theta\in D_\theta$, where the offline data $Z^N$ is generated from the dynamical system~\eqref{eq:general_dynamical_sysm}. Under Assumption~\ref{Assmp_Fisher} , we have that
	\begin{equation} \label{risk-lower-bound}
		R_\theta(\hat{h}(Z^N)) \geq \norm{b_\theta}_{W(\theta)}^2 + \Tr{W(\theta)m{'}(\theta)I_{F,N}^{-1}(\theta)(m{'}(\theta))^\top},
	\end{equation}
	where 
	\begin{equation} 
		m'(\theta) = \frac{\partial}{\partial \theta^\top}\mathbb{E}_\theta[\hat{h}(Z^N)]= b_\theta' + h'(\theta),
	\end{equation}
	where $b_\theta' \in \mathbb{R}^{n_h \times n_\theta}$ and $h' (\theta) \in \mathbb{R}^{n_h \times n_\theta}$ are the Jacobian of the bias $b_\theta$ and the optimal controller $h(\theta) $, respectively.
\end{theorem}

\begin{proof}
	The proof relies on CRLB and is found in Appendix~\ref{AppA_CRLB}.
\end{proof}

For convenience of notation, define 
\begin{equation*}
	f(\theta,b_\theta,b_\theta') :=\norm{b_\theta}_{W(\theta)}^2 + \Tr{W(\theta)m{'}(\theta)I_{F,N}^{-1}(\theta)(m{'}(\theta))^\top} .
\end{equation*}
For a class of decision rules with the same bias $b_\theta$, $f(\theta,b_\theta,b_\theta')$ itself guarantees a lower bound for the pointwise risk $R_\theta(\hat h)$. For instance, if the decision rule is unbiased as in \eqref{eq:unbiased_ddc}, the lower bound reduces to the unbiased CRLB (UCRLB)
\begin{equation} \label{risk-lower-bound_unbiased}
	R_\theta(\hat{h}) \geq \Tr{W(\theta)h{'}(\theta)I_{F,N}^{-1}(\theta)(h{'}(\theta))^\top}.
\end{equation}
In principle, to obtain the tightest possible risk bound, we would like to minimize $f(\theta,b_\theta,b_\theta')$ over all possible bias function $b_\theta$. For every fixed value of $\bar\theta$, taking $b_\theta=h(\bar\theta) - h(\theta)$ drives $f(\theta,b_\theta,b_\theta')$ to zero at $\theta=\bar\theta$. This is achieved at $\bar\theta$ by the decision $\hat{h} = h(\bar\theta)$. Clearly, this decision rule is not useful for values other than $\bar\theta$. Furthermore, $\hat{h}$ optimizes the bound for any specific $\theta=\bar\theta$ but not for all values of $\theta$. Thus, in general we cannot minimize $f(\theta,b_\theta,b_\theta')$ pointwisely for all $\bar\theta$. A meaningful minimization of \eqref{risk-lower-bound} is possible only after restricting attentions to a prescribed class of bias functions. For example, one may consider linear bias functions, as what was done in \cite{Eldar2008rethinking} for point estimation, and optimize the corresponding parameters within that class. Such a bound, however, remains tied to the chosen bias class and therefore does not capture the fundamental limitations of decisions with arbitrary bias. This motivates passing from the pointwise risk $R_\theta(\hat{h})$ to the average risk $r(\hat{h})$, for which a global lower bound can be established.

\vspace{-3mm}
\subsection{Optimal-Bias Lower Bound for the Average Risk}  \label{Sct5.2}

Based on \eqref{risk-lower-bound}, a lower bound on the average risk $r(\hat h)$ can be obtained by minimizing its right-hand side over all admissible bias functions. To this end, define the functional
\begin{equation}
	F(b_\theta)
	:=
	\int_{D_\theta} f(\theta,b_\theta,b_\theta')\,d\theta .
\end{equation}
Then the corresponding functional minimization problem is
\begin{equation} \label{Bayes-risk-CoV-scalar}
	\min_{b_\theta} \quad F(b_\theta).
\end{equation}
The solution to \eqref{Bayes-risk-CoV-scalar} provides a lower bound on the average risk $r(\hat h)$ for any data-driven controller $\hat h$. Here, $F(b_\theta)$ is the cost functional, and $f(\theta,b_\theta,b_\theta')$ is the Lagrangian, viewed as a function of three independent variables $\theta$, $b_\theta$, and $b_\theta'$. The functional minimization problem in \eqref{Bayes-risk-CoV-scalar} is precisely a calculus of variations problem, a framework that has long played a central role in both point estimation \cite{Young1971error,Ben2009lower} and optimal control \cite{Liberzon2011calculus}. To the best of our knowledge, however, this is the first time such a tool has been used to derive fundamental limitations in data-driven control problems. Solving the variational problem \eqref{Bayes-risk-CoV-scalar} yields a sharp CRLB-based lower bound on the average risk, and determines the optimal bias that attains this bound. As in point estimation, the optimal bias makes the bias-variance tradeoff in data-driven control transparent by showing how risk is balanced between the two terms.

In order to characterize the solutions of this functional minimization problem through optimality conditions, we first introduce some necessary differentiability assumptions.
\begin{assumption} \label{Assp1}
	The Lagrangian $f(\theta,b_\theta,b_\theta')$ is twice continuously differentiable of its three arguments $\theta$, $b_\theta$ and $b_\theta'$.
\end{assumption}

\begin{assumption} \label{Assp2}
	The bias $b_\theta$ is twice continuously differentiable.
\end{assumption}

\begin{remark}
	The above assumptions ensure that all derivatives appearing in the subsequent analysis are continuous. Although they can be further relaxed, to streamline the presentation we do not attempt to state the weakest possible regularity conditions. For related discussions, see \cite{Dacorogna2024introduction,Ben2009lower}.
\end{remark}

Since the Fisher information $I_{F,N}(\theta)$ and the weighting $W(\theta)$ are positive definite, the Lagrangian function $f(\theta,b_\theta,b_\theta')$ is strictly convex in $(b_\theta,b_\theta')$ for every $\theta \in \mathcal{D}_\theta$. Then, we have the following theorem regarding the minimizer of the variational problem \eqref{Bayes-risk-CoV-scalar}.
\begin{theorem}  \label{Theorem-EL-Equation-vector}
	Under Assumptions~\ref{Assmp_Fisher}, \ref{Assp1} and \ref{Assp2}, the minimizer $b_\theta^*$ of the variational problem \eqref{Bayes-risk-CoV-scalar} is unique, and satisfies the Euler-Lagrange equation
	\begin{equation}\label{eq: general form of EL equation-vector}
	\frac{\partial f}{\partial [b_\theta]_{k}}-\sum_{l=1}^{n_\theta}\frac{\partial }{\partial \theta_l}\left(\frac{\partial f}{\partial [b_\theta']_{k,l}} \right)=0,\ k \in \{1,\cdots,n_h\},
	\end{equation}
	and the natural boundary condition
	\begin{equation} \label{eq:natural boundary condition-vector}
		\frac{\partial f}{\partial b_\theta'} \bm{n} (\theta)=0,\ \forall \theta\in \partial {\mathcal{D}}_{\theta},
	\end{equation}
	where $[b_\theta]_{k}$ and $\theta_l$ denote the $k$-th and the $l$-th component of $b_\theta$ and $\theta$,  $[b_\theta']_{k,l}$ denotes the component on the $k$-th row and the $l$-th column of $b_\theta'$, and $\bm{n}(\theta)$ is the outward unit normal vector to the boundary $\partial {\mathcal{D}}_{\theta}$.
\end{theorem}
\begin{proof}
  See Appendix \ref{AppB_EL_Equation}.
\end{proof}

In general, the Euler-Lagrange equation~\eqref{eq: general form of EL equation-vector} in Theorem~\ref{Theorem-EL-Equation-vector} leads to a second-order ODE in $b_\theta$. In some simple settings the resulting ODE admits closed-form solutions, such as the solution \eqref{Gaussian-uniform-prior-lower-bound} for the point estimation in a ball \cite{Ben2009lower}; in most cases, it is solved numerically, e.g., using the MATLAB $\mathtt{pde}$ toolbox or Mathematica command $\mathtt{NDSolveValue}$. It should be emphasized that there is no guarantee that the solution $b_\theta^*$ of the Euler-Lagrange equation \eqref{eq: general form of EL equation-vector} is truly a bias function of a certain data-driven controller. We simply use it as an intermediary for the derivation of the lower bound.

We next apply the proposed approach to two data-driven control problems. This serves both to illustrate how the framework can be used in concrete settings and to assess the performance of several recently proposed data-driven controllers.
\vspace{-3mm}
\section{Case Study~I: Optimal Feedforward Control} \label{Sct6}

We begin with a simple feedforward control problem to illustrate the implications of our framework in a transparent setting. Although this example does not involve a dynamical system, it has bearings to state-of-the-art. In fact, the optimal feedforward control corresponds to one-step predictive control treated in \cite{Coulson2019data}. Moreover, its simplicity allows the fundamental limitations to be seen more clearly, thereby providing useful intuition for dynamical systems considered later.
\vspace{-3mm}
\subsection{Problem Formulation} \label{Sct6.1}

Consider a scalar model $y = \theta u + e$, where $e \sim \mathcal{N}(0, 1)$, and the domain of interest is given by $\mathcal{D}_\theta = [\theta_{\text{min}},\theta_{\text{max}}]$. Suppose that $r$ is the desired noise free output and that the objective is to design $u$ such that $y$ is close to $r$. However, we do not want to use a too large input $u$, we therefore use a quadratic cost to determine $u$, which is given by 
\begin{align*} \label{eq:feedforward_cost}
	J(u)= (r- y)^2 + \lambda u^2 = (r- \theta u)^2 + \lambda u^2,
\end{align*}
where the second term, which penalizes large inputs, can be tuned with the parameter $\lambda > 0$. The optimal controller is given by $u = h(\theta) = \frac{\theta}{\lambda + \theta^2}r$. Now define a quadratic cost for the data-driven controller, which is
 \begin{equation} \label{E5}
 	\begin{split}
 		L_{{\theta}}(\hat{h}) :=& J(\hat{h}) - J(h(\theta)) \\
        = &(r- \theta \hat{h})^2 + \lambda \hat{h}^2 - (r- \theta h(\theta))^2 - \lambda h^2(\theta) \\
 		= &W(\theta)\bigl(\hat h - h(\theta)\bigr)^2,
 	\end{split}	
\end{equation}
where we choose the weighting $W(\theta) = \lambda + \theta^2$. Based on the loss above, the corresponding risk $R_\theta(\hat h)$ and the average risk $r(\hat h)$ can be defined as in \eqref{risk} and \eqref{average_risk}, respectively. 

Based on design principles introduced in Section~\ref{Sct2}, we next develop several data-driven controllers for this problem, where data is given by the observation $z=y$, which is a realization of the random variable $Z \sim \mathcal{N}(\theta, 1)$. Moreover, the Fisher information $I_F = 1$.

\vspace{-3mm}
\subsection{Design of Data-Driven Feedforward Controllers}  \label{Sct6.2}

We consider four controllers for comparison. 

The first one is the CE controller $\hat{h}_{\text{CE}(\hat{\theta}_{\text{ML}})}$, obtained by replacing the true parameter $\theta$ in the optimal controller $h(\theta)$ with its ML estimator ${\hat \theta}_{\text{ML}} = z$, giving
\begin{equation} \label{Control-CEP}
	\hat{h}_{\text{CE}(\hat{\theta}_{\text{ML}})}(z) = h(\hat{\theta}_{\text{ML}}(z)) = \frac{z}{z^2+\lambda}r.
\end{equation}

The second one is the average risk tuning controller ${\hat h}_{A(P)}(z)$ in \eqref{eq:Bayes_control}, which is known to minimize the average risk $r(\hat h)$. To be specific, ${\hat h}_{A(P)}(z)$ is given by
\begin{equation}
   {\hat h}_{A(P)}(z) = \frac{ \int_{\theta_\text{min}}^{\theta_\text{max}} W(\theta)h(\theta) \exp\left( -\frac{(z-\theta)^2}{2} \right) d\theta }
	{ \int_{\theta_\text{min}}^{\theta_\text{max}} W(\theta) \exp\left( -\frac{(z-\theta)^2}{2} \right) d\theta }.
\end{equation}

The third one is the weighed average risk tuning controller ${\hat h}_{A(\pi)}(z)$, where the weighting $\pi(\theta)$ is the pdf of Gaussian distribution $\mathcal{N}(0,\tau^{-1})$, with $\tau > 0$. Then, it is easy to show that the controller ${\hat h}_{A(\pi)}(z)$ takes the form
\begin{equation} \label{Control-Pitman}
	\hat{h}_{A(\pi)}(z) = \frac{z}{z^2+\lambda+\frac{1}{1+\tau}}r,
\end{equation}
which corresponds to $h(\hat{\theta}_{\text{ML}}(z))$ where the penalty $\lambda$ in \eqref{Control-CEP} has been replaced by $\lambda+\frac{1}{1+\tau}$. We can interpret this as adding an information-dependent penalty on the control input: the less information is available, the larger the additional penalty is, i.e., a less certain control is more cautious.

Note that the controllers in \eqref{Control-CEP} and \eqref{Control-Pitman} both take the form
\begin{equation} \label{Control-structure}
\hat{h}_c(z) = \frac{z}{z^2+c}r,
\end{equation}
where $c \in \mathbb{R}$ is a hyperparameter to be tuned. A natural way to choose $c$ is to minimize the average risk of $\hat{h}_c(Z)$ over $D_\theta$ w.r.t. $c$, that is,
\begin{align*}
	\min_c \quad r(\hat h_c) := \int_{\theta_{\min}}^{\theta_{\max}} R_\theta(\hat h_c) d\theta.
\end{align*}
This minimization is carried out numerically and the resulting controller is used as the fourth controller, denoted by $\hat{h}_c$.

\vspace{-3mm}
\subsection{Calculating the Lower Bound for Feedforward Control}\label{Sct6.3}

We now derive the optimal-bias lower bound on the average risk $r(\hat h)$. Using the CRLB in Theorem~\ref{Thm:CRLB}, we have that 
\begin{align} \label{eq:LB_AV_feedforward}
	r(\hat h) \ge \int_{\theta_\text{min}}^{\theta_\text{max}} f\bigl(\theta,b_\theta,b_\theta'\bigr)d\theta,
\end{align}
where $f\bigl(\theta,b_\theta,b_\theta'\bigr) = b_\theta^2W(\theta)+W(\theta)(b_\theta' + \frac{\lambda - \theta^2}{(\lambda + \theta^2)^2}r)^2$. For this scalar case, the Euler-Lagrange equation \eqref{eq: general form of EL equation-vector} reduces to
\begin{equation} \label{eq: general form of EL equation-scalar}
	b_\theta W(\theta)-\frac{d}{d\theta}\left(W(\theta)\left(b_\theta'(\theta)+\frac{\lambda - \theta^2}{(\lambda + \theta^2)^2}r\right)\right) = 0,
\end{equation}
and the natural boundary condition \eqref{eq:natural boundary condition-vector} reduces to
\begin{equation}  \label{eq: natural_BC-scalar}
\begin{split}
   W(\theta)\left(b_\theta'(\theta)+\frac{\lambda - \theta^2}{(\lambda + \theta^2)^2}r\right)|_{\theta_{\min}}
	&= 0, \\ 
W(\theta)\left(b_\theta'(\theta)+\frac{\lambda - \theta^2}{(\lambda + \theta^2)^2}r\right)|_{\theta_{\max}}
	&= 0,
\end{split}
\end{equation}
respectively. By numerically solving the Euler-Lagrange equation \eqref{eq: general form of EL equation-scalar} subject to the boundary condition \eqref{eq: natural_BC-scalar}, we obtain both the lower bound for the average risk and the optimal bias that attains this bound. The results are shown in the sequel.
\vspace{-3mm}
\subsection{Simulation Results}\label{Sct6.4}

We now evaluate the performance of data-driven controllers introduced above against the lower bound for the average risk. We fix $\lambda = 5$, the reference $r=1$, and let $\theta$ vary over $\mathcal{D}_\theta = [0,2]$. For the third controller $\hat{h}_{A(\pi)}(z)$, we take $\tau = 1$ in the weighting $\pi(\theta)$. For the fourth controller $\hat{h}_c(z)$, using a numerical minimization procedure, we obtain the optimal value $c = 7.31$. For each fixed $\bar\theta$ in $\mathcal{D}_\theta$, we perform 100 Monte Carlo simulations to approximate the pointwise risk $R_\theta(\hat{h})$ for each controller.

\begin{figure}
	\centering
	\includegraphics[scale=0.45]{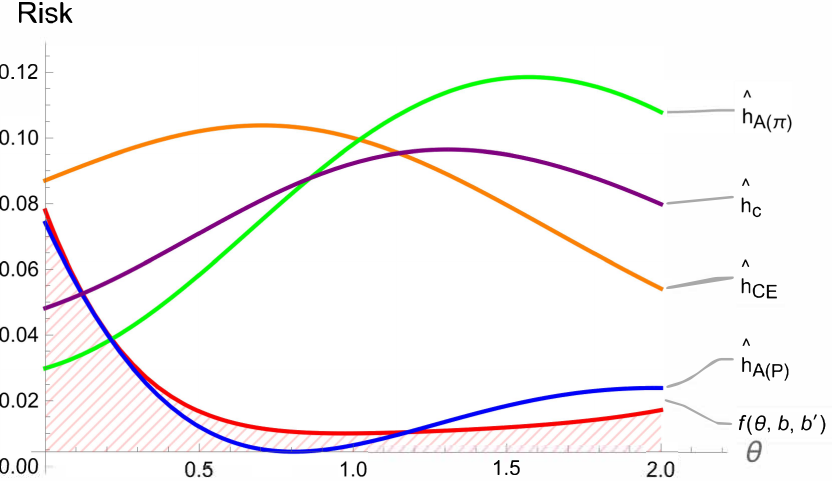}
	\caption{Risks of the four controllers. The average risks of $\hat{h}_{\text{CE}}$, $\hat{h}_{A(\pi)}$, $\hat{h}_{c}$ and $\hat{h}_{A(P)}$ are 0.178, 0.173, 0.163, and 0.039 respectively. The lower bound (hatched red region) is 0.038.}
	\label{fig:ex1-bayes-risks}
\end{figure}
Figure~\ref{fig:ex1-bayes-risks} compares the risks of four controllers with the derived lower bound for the average risk over $\mathcal{D}_\theta$. The red curve labeled $f(\theta,b_\theta,b_\theta')$ represents the value of the Lagrangian after substituting the optimal bias function $b_\theta^*$ and its derivative $(b_\theta^*)'$ into $f(\theta,b_\theta,b_\theta')$. As in the preview Figure~\ref{F1:MSE} of point estimation, Figure~\ref{fig:ex1-bayes-risks} shows that no single controller is uniformly best over the entire parameter domain $\mathcal{D}_\theta$. Instead, each design exhibits a different compromise between local and global performance. Among the four controllers, $\hat h_{A(P)}$ performs most closely to the theoretical lower bound. In particular, $\hat h_{A(P)}$ achieves the smallest average risk and is especially effective for moderate and large values of $\theta$, although its performance deteriorates as $\theta \to 0$. The remaining controllers display more pronounced departures from the lower bound. The controller $\hat h_c$ improves upon both $\hat h_{\text{CE}}$ and $\hat h_{A(\pi)}$ in terms of average risk, but still stays noticeably above the lower bound over much of the parameter space. Meanwhile, $\hat h_{\text{CE}}$ and $\hat h_{A(\pi)}$ each perform relatively better over some subregions of $\mathcal{D}_\theta$, at the price of substantially worse behavior elsewhere. Moreover, $f(\theta,b_\theta^*,(b_\theta^*)')$ should not be interpreted as a pointwise lower bound on the risk $R_\theta(\hat{h})$ but a global constraint. Therefore, there is no contradiction for a controller to have risk below $f(\theta,b_\theta^*,(b_\theta^*)')$ at some specific values of $\theta$ without violating the bound. For instance, the risk of average risk tuning controller $\hat h_{A(P)}$ is smaller than the lower bound on the subinterval $[0.5,1]$, yet when integrated over the total interval $[0,2]$, its average risk is still no smaller than the total area under the lower bound in \eqref{eq:LB_AV_feedforward}.

\begin{figure}
	\centering
	\includegraphics[scale=0.45]{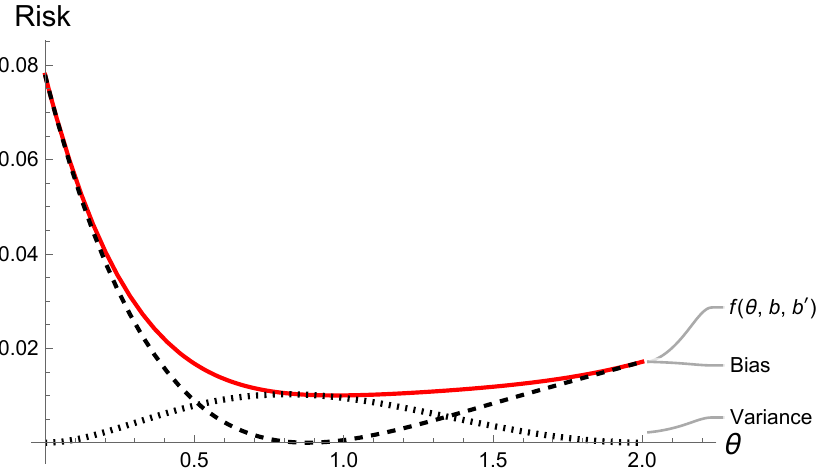}
	\caption{Bias-variance tradeoff in feedforward control}
	\label{fig:ex1-bayes-tradeoff}
\end{figure}

Figure~\ref{fig:ex1-bayes-tradeoff} shows the bias component $(b_\theta^*)^2W(\theta)$ and the surrogate variance component $W(\theta)\bigl((b_\theta^*)' + h'(\theta)\bigr)^2 I_F^{-1}$ appearing in $f(\theta,b_\theta^*,(b_\theta^*)')$. These two components exhibit the same rise-fall behavior observed in the point estimation example (see Figure~\ref{F1:MSE}). In regions where the variance lower bound decreases, the squared bias term increases, and vice versa. This clearly shows that there is also a  bias-variance tradeoff in data-driven control.

Before leaving this section, let us remark that, even if extremely simplified, this case study illustrates a fundamental trade-off in data-driven control: any reduction in risk relative to the lower bound over one region of the parameter space must be accompanied by an increase in risk elsewhere, analogous to the ``waterbed'' effect we introduced for point estimation. They also highlight the importance of evaluating how controller performance varies across the parameter space. Comparisons with several representative controllers show that the average risk minimization method is especially effective in approaching the lower bound, in line with recent Bayesian methods \cite{Chiuso2025harnessing,Schwaller2026bayesian}.

	\vspace{-3mm}
	\section{Case Study II: LQR} \label{Sct7}

We now show how our framework can be applied in data-driven control of dynamical systems. In particular, we consider data-driven LQR, a benchmark problem that has received considerable attention in recent years \cite{De2019formulas,De2021low,Dorfler2022bridging,Dorfler2023certainty,Zeng2025noise,Zhao2025data,Zhao2025regularization,Schwaller2026bayesian}.
    \vspace{-3mm}
	\subsection{Problem Formulation} \label{Sct7.1}
Given a discrete-time linear time-invariant (LTI) system
\begin{equation}\label{eq:LTI-system}
   x_{k+1} = A x_k + B u_k + e_k,
\end{equation}
where $x_{k}\in \mathbb{R}^{n_x}$, $u_{k}\in \mathbb{R}^{n_u}$, and $e_{k}\in \mathbb{R}^{n_x}$ are the system state, input and noise at time $k$, respectively. We assume that $(A, B)$ is stabilizable and $A$ is Schur stable.

The standard LQR problem for the system \eqref{eq:LTI-system} is
\begin{equation} \label{eq:LQR_problem}
\begin{split}
   \min_{\{u_k\}} \  & \lim_{T \to \infty} \mathbb{E}\left[\frac{1}{T} \sum_{k=0}^T (x_k^\top Qx_k + u_k^\top Ru_k)\right], \\
   \text{s.t.} \   & x_{k+1} = A x_k + B u_k + e_k,
\end{split}
\end{equation}
where the weighting matrices $Q \succ 0$, $R \succ 0$, and the expectation is over the randomness from the initial state $x_0$ and the noise $e_k$.  When system matrices $A$ and $B$ are known, this LQR problem \eqref{eq:LQR_problem} can be solved by finding the positive definite solution $P$ to the discrete-time algebraic Riccati equation
\begin{equation} \label{eq:LQR_Riccati}
   P = A^TPA - A^TPB (R + B^TPB)^{-1} B^TPA + Q.
\end{equation}
Then, the solution of \eqref{eq:LQR_problem} is $u_k = -K x_k$ with the optimal state feedback gain given by
\begin{equation} \label{eq:LQR_feedback_gain}
K = - (R + B^TPB)^{-1} B^TPA.
\end{equation}
Alternatively, following \cite{Feron1992Numerical} this optimal controller can be determined by solving the following program:
\begin{equation} \label{eq:LQR_SDP}
\begin{split}
   \min_{P \succcurlyeq I,K} \  & J(K) = \Tr{QP+K^\top RKP}, \\
   \text{s.t.} \   & (A+BK)P(A+BK)^\top - P + I \preccurlyeq 0,
\end{split}
\end{equation}
which can be further cast into a convex SDP after a change of variables; see \cite[Section III-C]{Dorfler2023certainty}.

The data-driven LQR problem considers the case that system matrices $A$ and $B$ are unknown, while a batch of data generated by the system is available. In this work, we assume that the data is generated as follows: $e_k \sim \mathcal{N}(0, \sigma_e^2 I_{n_x})$ are independent and identically distributed (i.i.d.), $u_k \sim \mathcal{N}(0, \sigma_u^2 I_{n_u})$ are i.i.d., and $x_0 \sim \mathcal{N}(0, \sigma_0^2 I_{n_x})$. In addition, $x_0$, $\{e_k\}$ and $\{u_k\}$ are mutually independent. Let us introduce the data matrices constructed from a trajectory collected over a horizon $N$:
\begin{subequations}\label{eq:LQR_data}
\begin{align}
X_{0}^N &= \begin{bmatrix}x_0&x_1&\cdots&x_{N-1}\end{bmatrix} \in \mathbb{R}^{n_x \times N}, \\
U_{0}^N &= \begin{bmatrix}u_0&u_1&\cdots&u_{N-1}\end{bmatrix} \in \mathbb{R}^{n_u \times N}, \\
X_{1}^N &= \begin{bmatrix}x_1&x_2&\cdots&x_{N}\end{bmatrix} \in \mathbb{R}^{n_x \times N}.
\end{align}
\end{subequations}

Before introducing the controller design, we formulate the data-driven LQR into the general framework in Section~\ref{Sct2}. To this end, we first introduce the following lemma.
\begin{lemma} [{\cite[Lemma 3]{Mania2019certainty}}]
Let $\hat K$ be an arbitrary stabilizing linear controller. Denote $\Sigma_{K} \succ 0$ and $\Sigma_{\hat K} \succ 0$ the covariance matrices of the stationary state of the closed-loop systems $A-BK$ and $A-B\hat K$, satisfying the Lyapunov equations
\begin{subequations}
\begin{align}
    \Sigma_{K} &= (A-BK)\Sigma_{K}(A-BK)^\top + \sigma_e^2I_{n_x}, \\
	\Sigma_{\hat K} &= (A-B\hat K)\Sigma_{\hat K}(A-B\hat K)^\top + \sigma_e^2I_{n_x},
\end{align}
\end{subequations}
respectively. Moreover, let $J(K)$ denote the cost in \eqref{eq:LQR_problem} corresponding to the optimal controller $K$ applied to the system $(A,B)$, and let $J(\hat K)$ denote the corresponding cost when the controller $\hat K$ is applied. Then, for suboptimality gap $\Delta_J(\hat K) =  J(\hat K) - J(K)$ we have
\begin{equation}  \label{eq:LQR_cost_difference}
	\Delta_J(\hat K) = \text{Tr}\left(\Sigma_{\hat K}(\hat{K}-K)^\top(R+B^\top PB)(\hat{K}-K)\right).
\end{equation}
\end{lemma}

The lemma above quantifies the suboptimality gap in terms of the controller mismatch $\hat{K}-K$. Note that \eqref{eq:LQR_cost_difference} can be rewritten as
\begin{equation}  \label{eq:LQR_cost_difference_rewritten}
	\begin{split}
 &\Delta_J(\hat K) =\text{Tr}\left(\Sigma_{K}(\hat{K}-K)^\top(R+B^\top PB)(\hat{K}-K)\right)  + \\ 
&\text{Tr}\left((\Sigma_{\hat K}-\Sigma_{K})(\hat{K}-K)^\top(R+B^\top PB)(\hat{K}-K)\right).
	\end{split}
\end{equation}
To apply our theory, we aim at a loss that is quadratic as in \eqref{loss}. According to \cite[Lemma 6]{Fazel2018global}, when $\hat{K}$ is sufficiently close to $K$, we have
\begin{equation} 
 	\Sigma_{\hat K} \approx \Sigma_{K} + \mathcal{O}(\norm{\hat{K}-K}),
\end{equation}
where $\norm{\hat{K}-K}$ denotes the spectral norm of $\hat{K}-K$. Then, $\text{Tr}\left(\Sigma_{K}(\hat{K}-K)^\top(R+B^\top PB)(\hat{K}-K)\right)$ is the leading-order term in \eqref{eq:LQR_cost_difference_rewritten}, and the remaining term is neglectable compared to it. Motivated by this, we define the following loss function as a surrogate of $\Delta_J(\hat K)$ in \eqref{eq:LQR_cost_difference_rewritten} for the data-driven LQR problem:
\begin{equation} \label{E5}
 	\begin{split}
 		L_{{\theta}}(\hat{h}) &= \text{Tr}\left(\Sigma_{K}(\hat{K}-K)^\top(R+B^\top PB)(\hat{K}-K)\right) \\
         &= \text{Vec}(\hat{K}-K)^\top \left(\Sigma_{K} \otimes (R+B^\top PB)\right)\text{Vec}(\hat{K}-K) \\
		 &= (\hat{h}-h(\theta))^\top W(\theta)(\hat{h}-h(\theta)),
 	\end{split}	
\end{equation}
where $\theta = \text{Vec}(\begin{bmatrix}A & B\end{bmatrix}) \in \mathbb{R}^{n_x(n_x+n_u)}$ denotes the vectorization of matrices $A$ and $B$ by column, $h(\theta) = \text{Vec}(K)\in \mathbb{R}^{n_xn_u}$, and $\hat h = \text{Vec}(\hat K)\in \mathbb{R}^{n_xn_u}$. Moreover, $W(\theta) = \Sigma_{K} \otimes (R+B^\top PB) \in \mathbb{R}^{n_xn_u \times n_xn_u}$ is the Kronecker product of $\Sigma_{K}$ and $(R+B^\top PB)$. Based on the above quadratic loss taking the role of \eqref{loss}, the pointwise risk $R_\theta(\hat h)$ and average risk $r(\hat h)$ are analogously defined as in \eqref{risk} and \eqref{average_risk}, respectively. It should be mentioned that the risk $R_\theta(\hat h)$ quantifies the expected long-run performance degradation of a data-driven controller $\hat h$. This differs from the viewpoint commonly adopted in some literature \cite{De2019formulas,Van2020data,De2021low,Dorfler2023certainty}, where guarantees are typically expressed directly in terms of a realized $\Delta_J(\hat K)$. 

Having formulated the data-driven LQR problem within the framework of Section~\ref{Sct2}, we now turn to the design of data-driven LQR controllers.
\vspace{-3mm}
\subsection{Design of Data-Driven LQR} \label{Sct7.2}

By the linear dynamics \eqref{eq:LTI-system}, system matrices can be expressed in a linear regression form
\begin{equation}\label{eq:LQR_linear_regression}
	X_{1}^N = A X_{0}^N + B U_{0}^N + E_{0}^N = \begin{bmatrix} B & A \end{bmatrix} Z^N + E_{0}^N,
\end{equation}
where $Z^N = [(U_{0}^N)^\top \ (X_{0}^N)^\top]^\top$, and $E_{0}^N$ stacks the noise sequence, defined analogously to the data matrices in \eqref{eq:LQR_data}. Throughout this section, we assume that the data matrix $Z^N$ is full row rank, which is a necessary condition for data-driven LQR \cite{Van2020data}.

We first consider the indirect CE controller. The usual ML estimator for system matrices is given by
\begin{equation} \label{eq:MLE}
\hat\theta_{\text{ML}} := {\begin{bmatrix}\hat B &\hat A \end{bmatrix} }_{\text{ML}} = X_1^N (Z^N)^\top \bigl(Z^N (Z^N)^\top\bigr)^{-1}.
\end{equation}
By replacing the unknown parameter $\theta$ in the Riccati equation \eqref{eq:LQR_Riccati} with $\hat\theta_{\text{ML}} $, solving for the corresponding matrix $\hat P$, and then substituting $\hat P$ into \eqref{eq:LQR_feedback_gain}, we obtain the CE controller $\hat h_{\text{CE}(\hat\theta_{\text{ML}})} = h(\hat\theta_{\text{ML}})$. According to \cite{Hjalmarsson2005experiment}, given that $\hat\theta_{\text{ML}}$ is consistent and asymptotically efficient, $h(\hat\theta_{\text{ML}})$ is also asymptotically efficient. Moreover, recent works have provided finite-sample analysis for $\hat\theta_{\text{ML}}$ \cite{Jedra2022finite,Simchowitz2018learning} and $h(\hat\theta_{\text{ML}})$ \cite{Mania2019certainty,Dean2020sample}. 

The second family of controllers we introduce is the so called direct data-driven LQR methods, which use data to synthesize a controller without explicitly identifying system matrices. There are many variants of direct data-driven LQR in the recent literature; see, for example, \cite{De2019formulas,De2021low,Van2020data,Van2020noisy,Dorfler2023certainty,Zhao2025regularization,Schwaller2026bayesian}. We here highlight several representative milestones for comparison. 

One of the earliest methods in this line of work appears in \cite{De2019formulas,Van2020data}, which is able to recover the optimal gain $K$ exactly in the absence of noise. However, recent work \cite[Th.~2]{Zeng2025noise} shows that the resulting controller is not a consistent estimator of $K$ under noisy data unless suitably regularized with a data-size dependent regularization coefficient. It was later shown in \cite{Dorfler2023certainty} that the indirect and direct methods can be viewed within a unified framework, which leads to the following regularized formulation:
 \begin{equation} \label{eq:LQR_SDP_Reg}
    \begin{array}{cl}
        \min_{X, Y} & \Tr{Q X_0^N Y} + \Tr{X} + \lambda_1 \norm{\Pi Y}, \\
        \text {s.t. } & 
        \left[\begin{array}{cc}
        X_{0}^N Y - I_{n_x} &  X_{1}^N  Y \\
        Y^{\top} X_{1}^{\top} & X_{0}^N Y
        \end{array}\right] \succcurlyeq 0\\
        &\left[\begin{array}{cc}
        X &  R^{\frac{1}{2}}U_0^N Y  \\
        \left(R^{\frac{1}{2}} U_0^N Y \right)^{\top} & X_0^N Y
        \end{array}\right] \succcurlyeq 0,
    \end{array}
\end{equation}
where $\Pi = I_N - (Z^N)^\dagger Z^N$. Let $Y_R$ be the optimal solution of \eqref{eq:LQR_SDP_Reg}. The resulting controller is then given by
\begin{equation}\label{eq:LQR_SDP_Reg_sol}
\hat{h}_{R} := \hat{K}_{R} = -U_0^N Y_R (X_0^N Y_R)^{-1}.
\end{equation} 
According to \cite{Dorfler2023certainty}, the regularization term $ \lambda_1 \norm{\Pi Y}$ reflects a least-squares data-fitting criterion. When $\lambda_1 = 0$, \eqref{eq:LQR_SDP_Reg} reduces to the method in \cite[Th.~4]{De2019formulas}; when the regularization parameter $\lambda_1$ is sufficiently large, it coincides with $h(\hat\theta_{\text{ML}})$. 

Recently, a covariance-parameterized LQR formulation was proposed in \cite{Zhao2025data}, followed by a regularized variant in \cite{Zhao2025regularization}, where the regularization term accounts for uncertainty in both the steady-state covariance and the LQR objective. Very recently, using the same covariance parameterization as in \cite{Zhao2025regularization}, a direct Bayesian LQR method was proposed in \cite{Schwaller2026bayesian} following the procedure in \cite{Chiuso2025harnessing}. In this approach, posterior uncertainty is propagated into the control design through a variance-based regularization term. Specifically, system matrices $B$ and $A$ are assigned a matrix-normal prior $\pi$ with mean $\bar B$ and $\bar A$ and covariance $\Omega^{-1}$. Given the data $Z^N$, the posterior distribution of $A$, $B$ and the closed-loop $A - B\hat{K}$ remains matrix normal \cite[Lemma~1]{Schwaller2026bayesian}. This leads to the following SDP formulation \cite[Proposition~1]{Schwaller2026bayesian}:
\begin{equation}\label{eq:cov_convex_form_Bayes}
\begin{aligned}
\min_{\Sigma,S,L,M}\quad
& \Tr{Q\Sigma} + \Tr{RL} + \lambda_2 \Tr{M\Psi}, \\
\text{s.t.}\quad
& \Psi_2 S = \Sigma, \ \begin{bmatrix}
\Sigma - \sigma_e^2I_{n_x} & \bar X_1 S \\
S^\top \bar X_1^\top & \Sigma
\end{bmatrix} \succcurlyeq 0, \\
& \begin{bmatrix}
L & \Psi_1 S \\
S^\top \Psi_1^\top & \Sigma
\end{bmatrix} \succcurlyeq 0, \ \begin{bmatrix}
M & S \\
S^\top & \Sigma
\end{bmatrix} \succcurlyeq 0,
\end{aligned}
\end{equation}
where
\begin{equation}
  \begin{split}
     \Psi
     &=
     \frac{Z^N(Z^N)^\top + \sigma_e^2\Omega}{N}
     =
     \begin{bmatrix}
     \Psi_1\\
     \Psi_2
     \end{bmatrix}, \\
     \bar{X}_1
     &=
     \frac{X_1^N(Z^N)^\top + \sigma_e^2[\bar{B} \ \bar{A}]\Omega}{N}.
  \end{split}
\end{equation}
Here, $\Psi_1\in\mathbb{R}^{n_u\times(n_x+n_u)}$ and
$\Psi_2\in\mathbb{R}^{n_x\times(n_x+n_u)}$ are state-input covariances. Let $S_\pi$ and $\Sigma_\pi$ denote the optimal solutions to \eqref{eq:cov_convex_form_Bayes}. The resulting controller is then given by
\begin{equation}\label{eq:LQR_SDP_COV_Reg_sol}
\hat{h}_{A(\pi)}
:=\hat K_\pi=-\Psi_1 S_\pi \Sigma_\pi^{-1}.
\end{equation}

Compared with \eqref{eq:LQR_SDP_Reg}, the formulation \eqref{eq:cov_convex_form_Bayes} has decision variables whose dimensions are independent of the data length $N$, and therefore enjoys better computational scalability. Moreover, numerical simulations in \cite{Schwaller2026bayesian} indicate that the Bayesian direct LQR method \eqref{eq:cov_convex_form_Bayes} achieves a lower median optimality gap and a higher closed-loop stability rate than existing approaches.

The third type of controllers we consider is the average risk tuning controller ${\hat h}_{A(P)}(Z^N)$ in \eqref{eq:Bayes_control}, which is given by
\begin{equation} \label{eq:Bayes_control_LQR}
    \hat h_{A(P)}(Z^N) =\overline W^{-1}(Z^N) \overline H(Z^N),
\end{equation}
where 
\begin{equation*}
	\begin{split}
		\overline W(Z^N) &=\int_{\mathcal{D}_\theta} W(\theta)p(Z^N;\theta)d\theta, \\
		\overline H(Z^N) &=\int_{\mathcal{D}_\theta} W(\theta)h(\theta)p(Z^N;\theta)d\theta.
	\end{split}
\end{equation*}
The pdf $p(Z^N;\theta)$ can be decomposed using the Markovian property
\begin{equation}
   p(Z^N;\theta) = p(x_0) \prod_{k=0}^{N-1} p(x_{k+1} | x_k, u_k; \theta),
\end{equation}
where the conditional distribution of $x_{k+1}$ is
\begin{equation}
   (x_{k+1} | x_k, u_k; \theta) \sim \mathcal{N}(A x_k + B u_k, \sigma_e^2 I_{n_x}).
\end{equation}
In the implementation, the integrals $\overline W(Z^N)$ and $\overline H(Z^N)$ are approximated using  grid-based integration over the feasible set $\mathcal{D}_\theta$. To be specific, a fixed grid of candidate models $(A_i,B_i)$ is constructed within $\mathcal{D}_\theta$. For each valid grid point, we precompute the corresponding LQR gain $K_i$ and weighting term $W_i$. Given the observed data $Z^N$, the likelihood $p(Z^N;\theta_i)$ is evaluated for each candidate model $\theta_i=(A_i,B_i)$, and the average-risk controller is computed as the resulting likelihood-weighted average. 

The controllers $\hat h_{A(\pi)}$ in \eqref{eq:LQR_SDP_COV_Reg_sol} and $\hat h_{A(P)}$ in \eqref{eq:Bayes_control_LQR} can be compared within the average risk minimization (Bayes) framework introduced in Section~\ref{Sct2.1}. This comparison is postponed to the end of this section, where simulation results are presented to visualize the effect of the prior choice.
\vspace{-3mm}
\subsection{Calculating the Lower Bound for LQR} \label{Sct7.3}

We now apply the methodology developed in Section~\ref{Sct5} to derive lower bounds for the risk $R_\theta(\hat h)$ in \eqref{risk} and the average risk $r(\hat h)$ in \eqref{average_risk} for data-driven LQR. For the pointwise risk $R_\theta(\hat h)$, we have
\begin{equation} \label{lower_bound_point_wise_LQR}
    R_\theta(\hat h) \ge f(\theta,b_\theta,b_\theta'),
\end{equation}
where $f(\theta,b_\theta,b_\theta')$ is defined in \eqref{risk-lower-bound}. In particular, for unbiased controllers, the lower bound \eqref{lower_bound_point_wise_LQR} reduces to the UCRLB~\eqref{risk-lower-bound_unbiased}.

For the average risk, we have
\begin{align}
    r(\hat h) \ge \int_{\mathcal{D}_\theta} f(\theta,b_\theta,b_\theta')d\theta
    =: F(b_\theta).
\end{align}
The optimal-bias lower bound for $r(\hat h)$ is thus obtained by solving the variational problem of minimizing $F(b_\theta)$.

To compute the lower bounds above, one needs the Fisher information matrix $I_{F,N}(\theta)$, the optimal controller $h(\theta)$, and its derivative $h'(\theta)$. For the scalar case with open-loop data, detailed expressions of the relevant quantities are provided in Appendices~\ref{AppC_Fisher} and~\ref{AppD_CoV_LQR}, respectively. Once these quantities are available, the optimal-bias lower bound for the average risk $r(\hat h)$ can be computed numerically by solving the the Euler-Lagrange equation \eqref{eq: general form of EL equation-vector} subject to the natural boundary conditions \eqref{eq:natural boundary condition-vector}.  The multivariable case can be numerically solved in a similar manner. 

\vspace{-3mm}
\subsection{Simulation Results} \label{Sct7.4}

We now use scalar systems to illustrate the implications of the theory and to benchmark the performance of the controllers introduced above. Without loss of generality, we consider the domain of interest $\mathcal{D}_\theta = \left\{\theta : (A-A_\circ)^2 + (B- B_\circ)^2 \le \rho_\circ^2 \right\}$, where $A_\circ = B_\circ = 0.5$ and $\rho_\circ = 0.4$. The experimental setting is as follows: $Q = 1$, $R = 0.1$, sample size $N = 20$, and $\sigma_e^2 = \sigma_0^2 = \sigma_u^2 = 1$. For each grid point $\theta_i = (A_i, B_i)$ in $\mathcal{D}_\theta$, we perform 200 Monte Carlo trials and use their average to estimate the risk and suboptimality gap of each realized data-driven LQR controller. For the regularization parameters in $\hat h_R$ and $\hat h_{A(\pi)}$, we take $\lambda_1 = 1$ and $\lambda_2 = 0.05$. Moreover, for the Gaussian prior $\pi$ used in $\hat h_{A(\pi)}$, we set $\bar A=\bar B=0.5$ and $\Omega^{-1}=\operatorname{diag}(\sigma_A^2,\sigma_B^2)$ with $\sigma_A^2=\sigma_B^2=0.01778$, so that the marginal values $A=B=0.1$ and $A=B=0.9$ lie approximately within the $3\sigma$ confidence interval centered at $\bar A=\bar B=0.5$. Two fundamental lower bounds are computed for illustration. The first is the UCRLB \eqref{risk-lower-bound_unbiased} for the pointwise risk $R_\theta(\hat h)$ of unbiased decision rules, and the second is the optimal-bias lower bound on the average risk $r_\theta(\hat h)$ that applies to all controllers.

\begin{figure}
    \centering
    \includegraphics[scale=0.16]{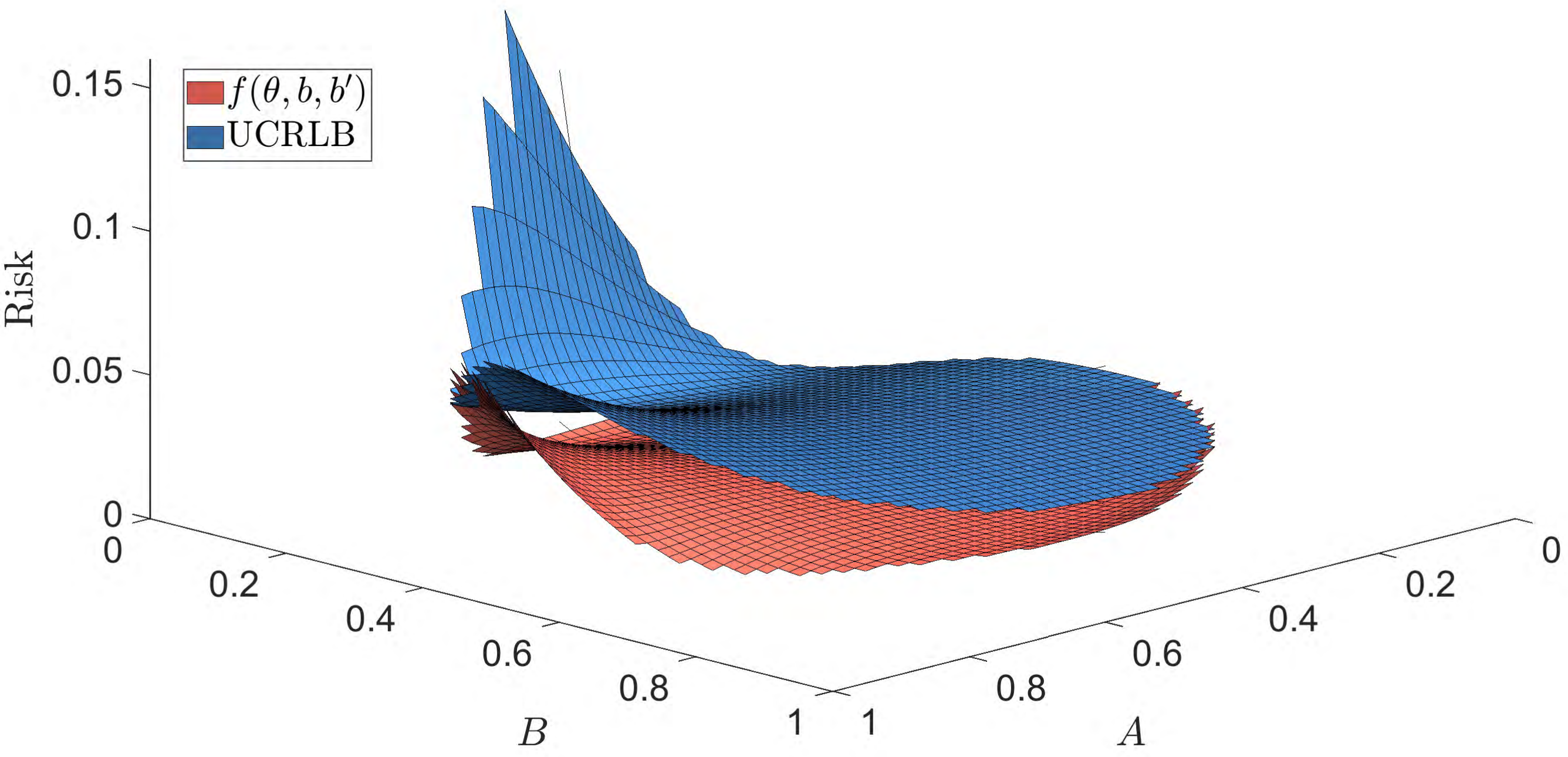}
    \caption{Comparison of the two lower bounds.}
    \label{fig:ex3-two-bounds}
\end{figure}

The resulting two lower bounds over $\mathcal{D}_\theta$ are given in Figures~\ref{fig:ex3-two-bounds}. The red surface labeled $f(\theta,b_\theta,b_\theta')$ represents the value of the Lagrangian after substituting the optimal bias function $b_\theta^*$ and its derivative $(b_\theta^*)'$ into $f(\theta,b_\theta,b_\theta')$, and the blue surface labeled UCRLB represents the bound \eqref{risk-lower-bound_unbiased} for the pointwise risk $R_\theta(\hat h)$. The lower bounds become significantly larger when $B$ is small and $A$ is close to $1$, indicating a region in which good control performance from finite data is intrinsically hard. This is especially pronounced for the UCRLB, which rises sharply in that part of the parameter domain. To understand this behavior, we examine the Fisher information. As shown in the Appendix, under the assumption that the system is asymptotically stable, i.e., $|A|<1$, then the state variance converges to a steady-state value $\sigma_\infty^2$ satisfying
\begin{align}
    \sigma_\infty^2 = A^2 \sigma_\infty^2 + B^2 \sigma_u^2 + \sigma_e^2,
\end{align}
which yields $\sigma_\infty^2 = \frac{B^2 \sigma_u^2 + \sigma_e^2}{1-A^2}$. Consequently, the Fisher information grows linearly with $N$ and is approximately given by $I_{F,N}(\theta) \approx \frac{N}{\sigma_e^2}{\rm{diag}}\left(\dfrac{B^2 \sigma_u^2 + \sigma_e^2}{1-A^2},\sigma_u^2\right)$. This expression shows that, as $B$ becomes smaller and $|A|$ becomes larger, the Fisher information $I_{F,N}(\theta)$ decreases, indicating that the system becomes increasingly difficult to identify. Since the lower bounds depend on the inverse of Fisher information matrix, they increase significantly in this regime. This leads to an important observation: systems that are difficult to identify are likewise difficult to control in a data-driven manner. An alternative interpretation is that when $B$ becomes smaller and $|A|$ becomes larger, the system loses stabilizability. Our observation therefore is consistent with the arguments in \cite{Tsiamis2022learning} and \cite{Ziemann2024regret}: systems with poor stabilizability are intrinsically hard to learn to control.

\begin{figure}
	\centering
	\includegraphics[scale=0.16]{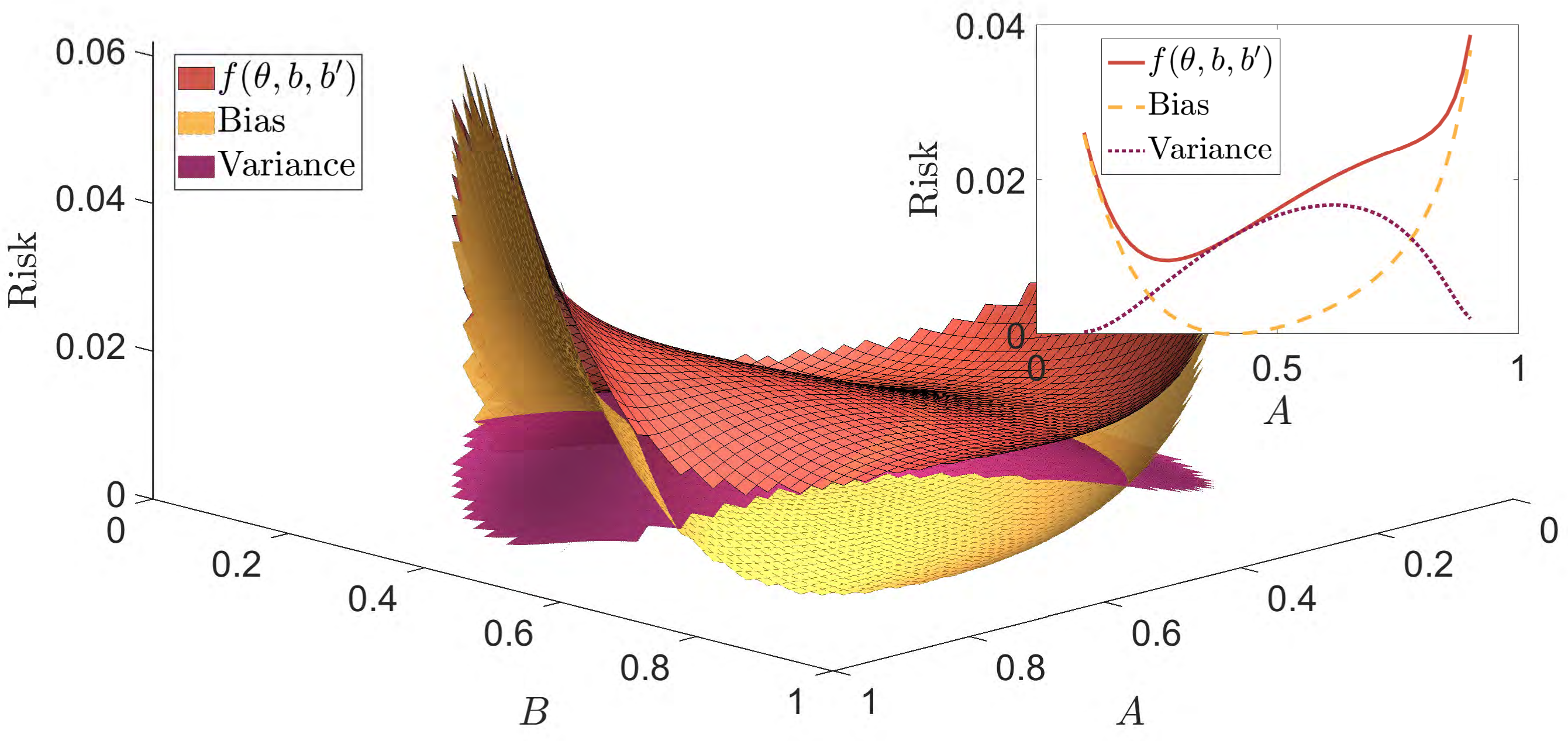}
	\caption{The bias-variance tradeoff in LQR, with a slice at $B = 0.5$ along the direction $A \in [0.1,0.9]$.}
	\label{fig:ex3-tradeoff}
\end{figure}

Figure~\ref{fig:ex3-tradeoff} shows the squared bias and the surrogate variance component independently appearing in $f(\theta,b_\theta^*,(b_\theta^*)')$, which exhibit the similar rise-fall trade-off behavior observed in the previous examples.

Figure~\ref{fig:ex3-all-risks} compares risks of the four data-driven LQR controllers introduced in Section~\ref{Sct7.2}. It is worth noting that, unlike the lower bounds in Figures~\ref{fig:ex3-two-bounds}, which vary smoothly over the parameter domain, the empirical pointwise risks and suboptimality gaps of the data-driven controllers exhibit visible fluctuations. This is an expected phenomenon, caused by the use of a finite data length and a limited number of Monte Carlo runs in estimating the risk, and further amplified by the nonlinear mapping from the estimated parameters to the resulting controller and cost. Despite these local irregularities, the overall trends of the empirical risks remain consistent with the derived lower bounds (to be seen later). Figure~\ref{fig:ex3-all-risks} conveys two main messages. First, it supports the argument that systems that are difficult to identify are also difficult to control in a data-driven manner: the risks of all controllers increase significantly in the region where $B \to 0$ and $A \to 1$. Second, it illustrates that there is no uniformly superior data-driven LQR controller over the entire parameter domain. Among the four methods, although $\hat{h}_{A(P)}$ generally achieves the smallest risk over a substantial portion of the parameter space, reflecting its average risk minimization design, there remain nontrivial subregions where the other three controllers perform better. For instance, $\hat h_{A(\pi)}$ performs better than $\hat h_{A(P)}$ near the center of the domain. Hence, a controller that is optimal in an average sense need not be optimal pointwise. The figure therefore demonstrates that the relative performance of data-driven LQR controllers depends strongly on the underlying system parameters. 

\begin{figure}
	\centering
	\includegraphics[scale=0.16]{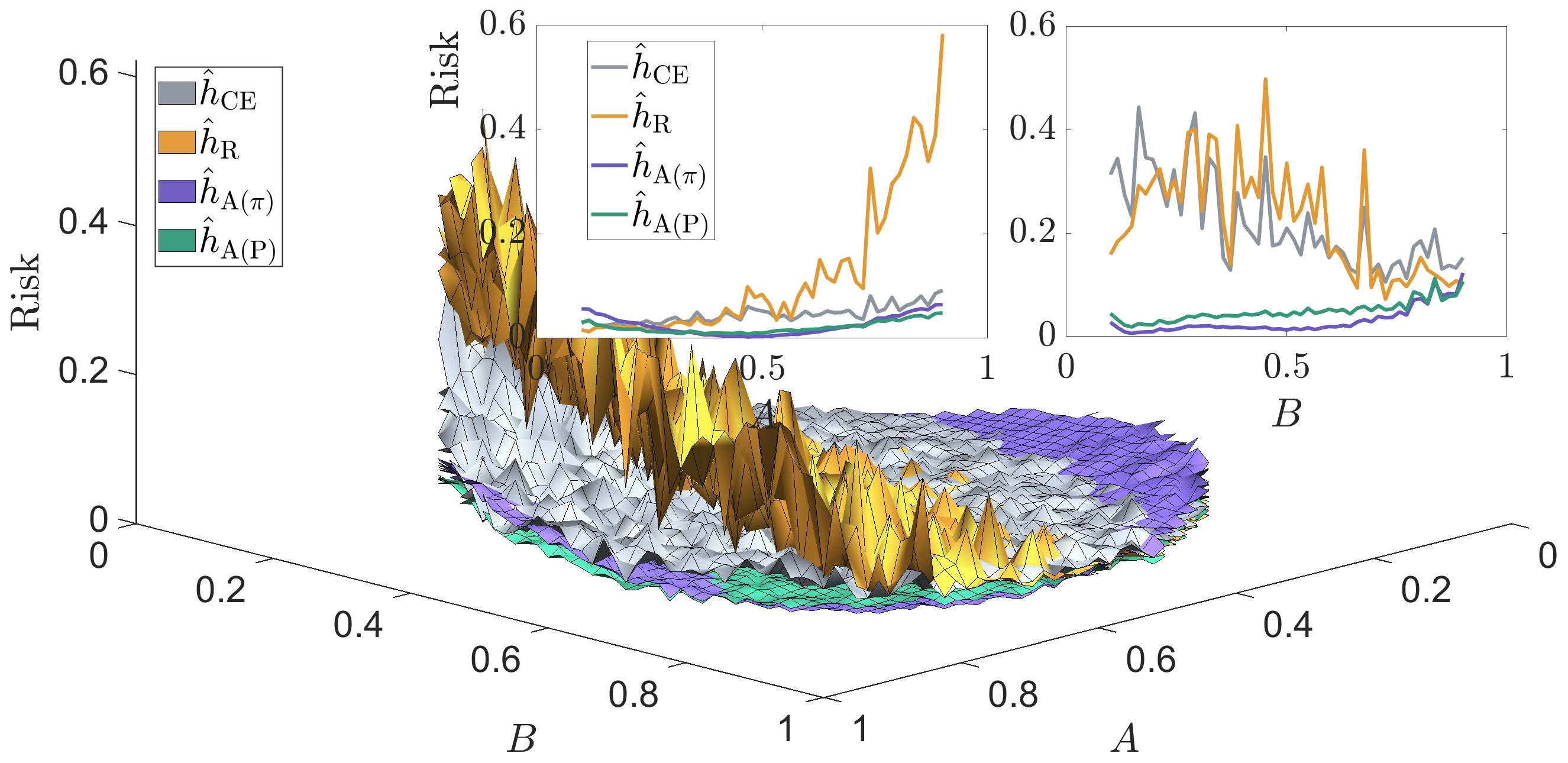}
	\caption{Pointwise risks $R_\theta (\hat{h})$ of different controllers over $\mathcal{D}_\theta$. The two slices correspond to varying $A\in[0.1,0.9]$ with $B=0.5$, and varying $B\in[0.1,0.9]$ with $A=0.5$.}
	\label{fig:ex3-all-risks}
\end{figure}
\begin{figure}
	\centering
	\includegraphics[scale=0.16]{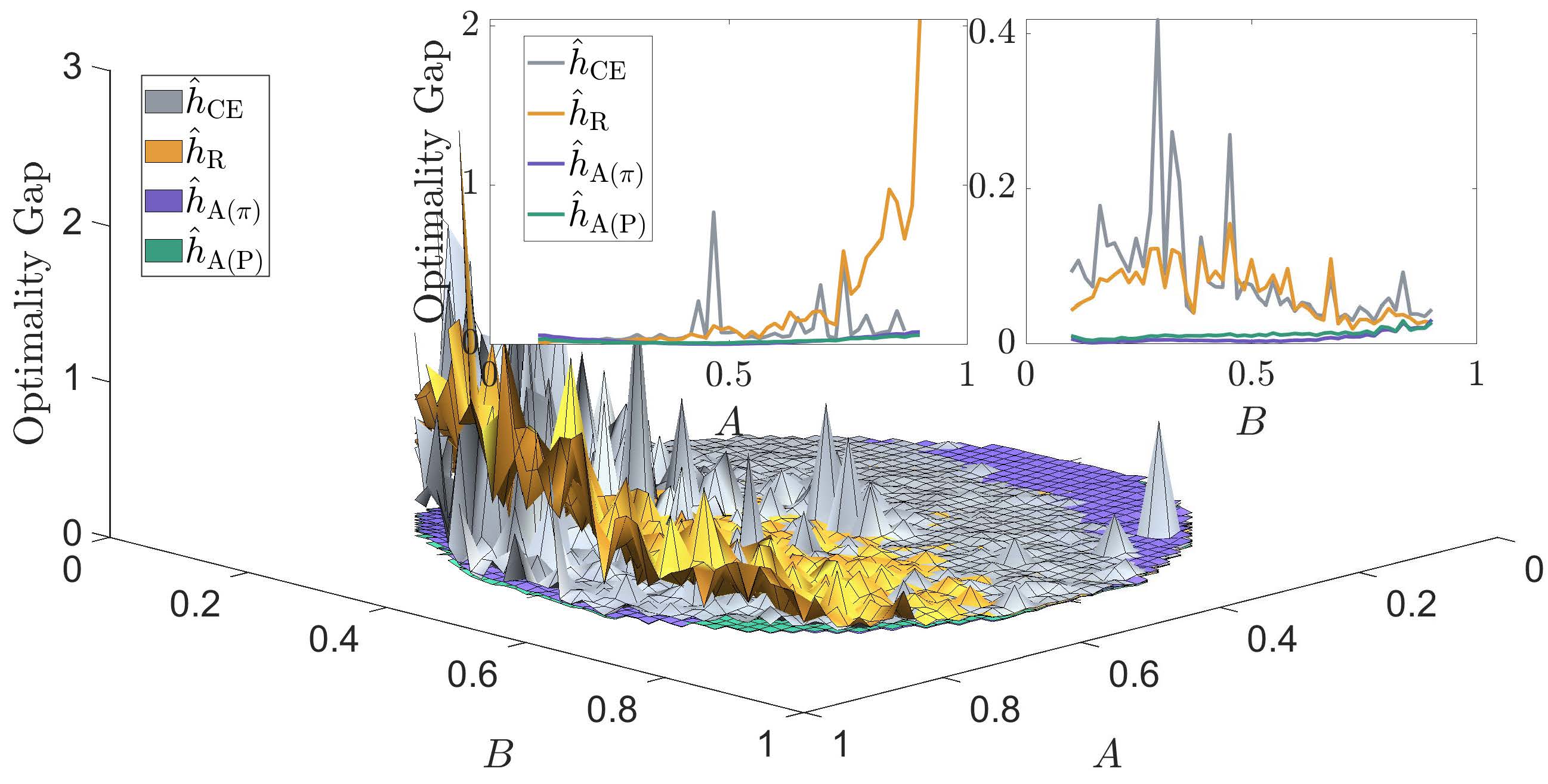}
	\caption{Suboptimality gaps of different controllers over $\mathcal{D}_\theta$. The two slices correspond to varying $A\in[0.1,0.9]$ with $B=0.5$, and varying $B\in[0.1,0.9]$ with $A=0.5$.}
	\label{fig:ex3-all-gaps}
\end{figure}

Figure~\ref{fig:ex3-all-gaps} compares the realized suboptimality gaps of the four data-driven LQR controllers. As shown in the figure, the suboptimality gaps exhibit behavior similar to that of the risks in Figure~\ref{fig:ex3-all-risks}, supporting the use of the loss function \eqref{E5} as a meaningful surrogate for the realized suboptimality gap $\Delta_J(\hat K)$ \eqref{eq:LQR_cost_difference} for the data-driven LQR problem.

\begin{figure}[!t]
	\centering
	
	\begin{subfigure}[t]{\columnwidth}
		\centering
		\includegraphics[width=0.82\linewidth]{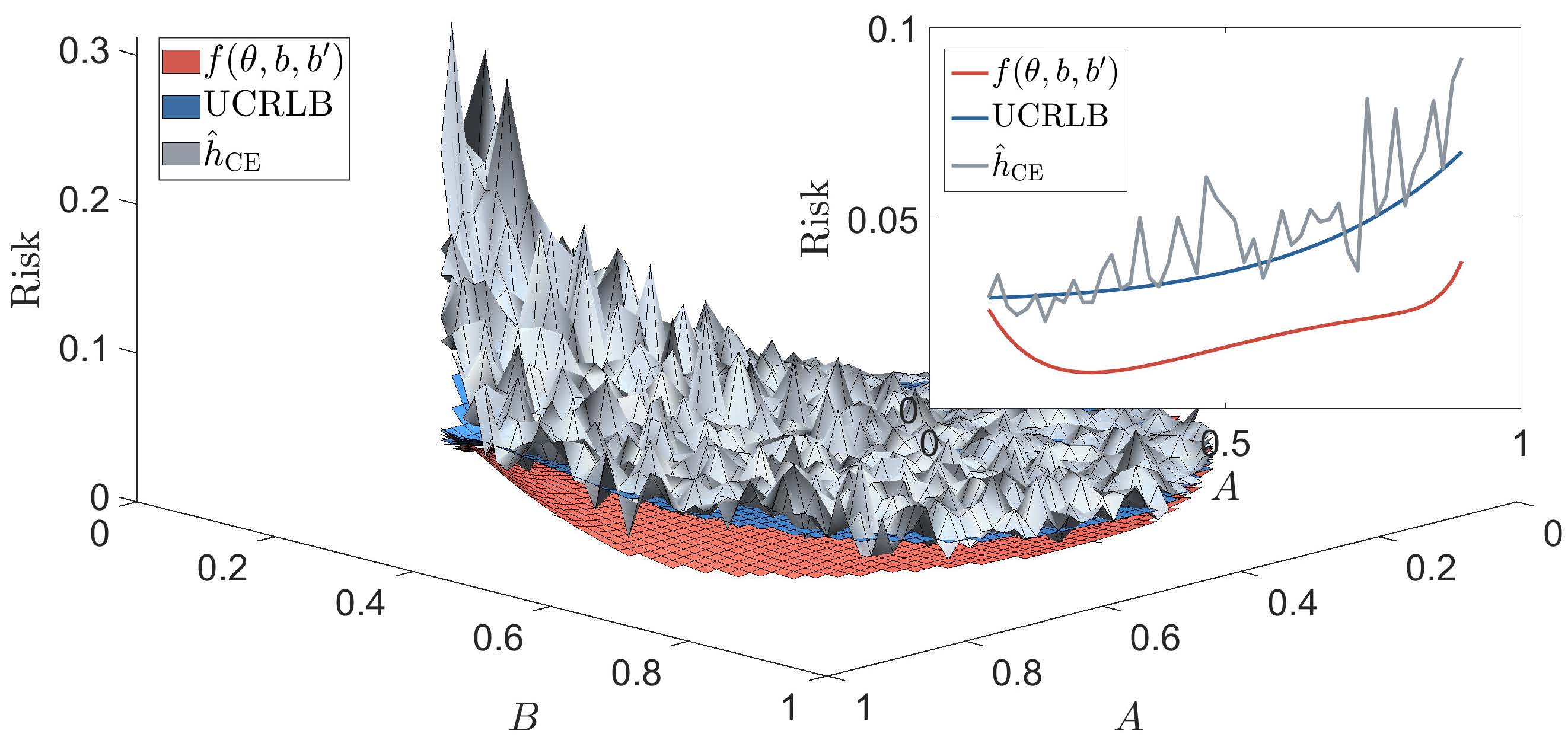}
		\caption{Pointwise risk $R_\theta (\hat{h})$ of the certainty equivalence controller $\hat{h}_{\text{CE}}$}
		\label{fig:cep}
	\end{subfigure}
	
	\vspace{0.3em}
	
	\begin{subfigure}[t]{\columnwidth}
		\centering
		\includegraphics[width=0.82\linewidth]{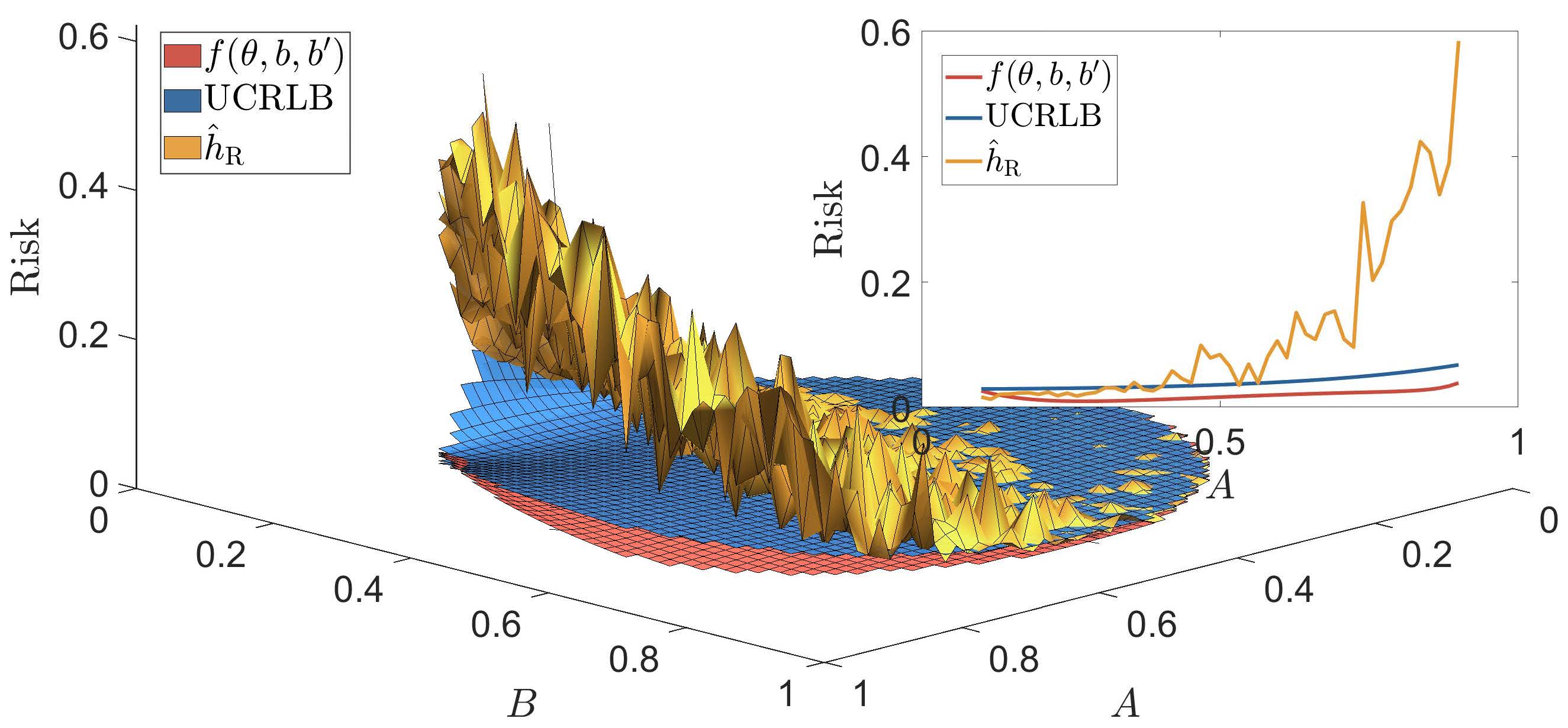}
		\caption{Pointwise risk $R_\theta (\hat{h})$ of the SDP-based direct method $\hat{h}_{R}$  \cite{Dorfler2023certainty}}
		\label{fig:sdp}
	\end{subfigure}
	
	\vspace{0.3em}
	
	\begin{subfigure}[t]{\columnwidth}
		\centering
		\includegraphics[width=0.82\linewidth]{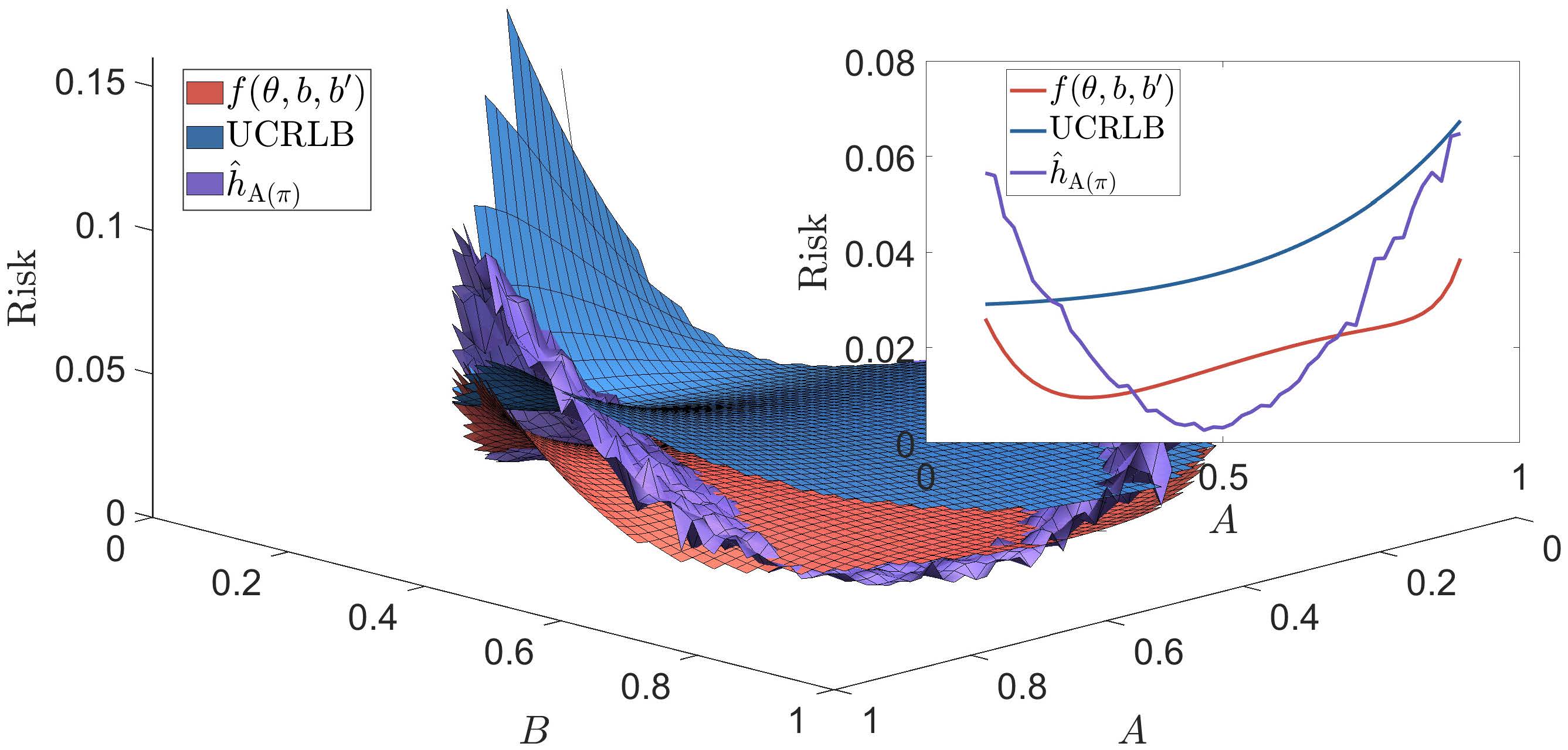}
		\caption{Pointwise risk $R_\theta (\hat{h})$ of the covariance-parameterized Bayesian method $\hat{h}_{A(\pi)}$ \cite{Schwaller2026bayesian}}
		\label{fig:sdpr}
	\end{subfigure}
	
	\vspace{0.3em}
	
	\begin{subfigure}[t]{\columnwidth}
		\centering
		\includegraphics[width=0.82\linewidth]{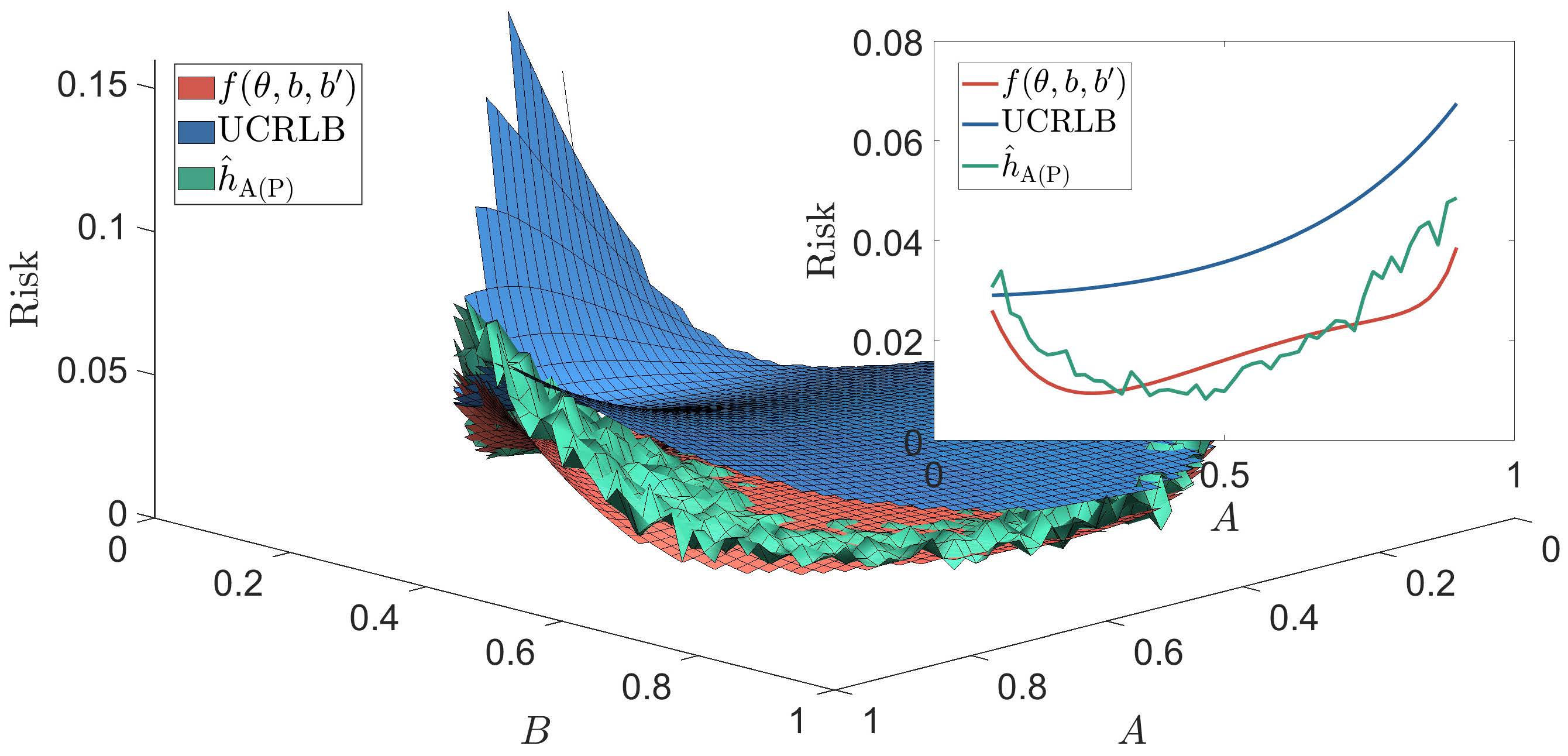}
		\caption{Pointwise risk $R_\theta (\hat{h})$ of the average risk tuning controller $\hat{h}_{A(P)}$}
		\label{fig:ap}
	\end{subfigure}
	
	\caption{Comparison of data-driven LQR controllers vs. the lower bounds over $\mathcal{D}_\theta$, with slices at $B = 0.5$ along the direction $A \in [0.1,0.9]$.}
	\label{fig:LQR_comparisons}
\end{figure}

The four subplots in Figure~\ref{fig:LQR_comparisons} compare the performance of four data-driven LQR controllers with the corresponding lower bounds. In Figure~\ref{fig:cep}, the certainty equivalence controller $\hat h_{\text{CE}}$ follows a similar profile and remains relatively close to the UCRLB over much of the parameter domain. Since this controller is consistent \cite{Mania2019certainty}, the UCRLB remains a meaningful pointwise lower bound for the risk of $\hat h_{\text{CE}}$ and serves as a useful benchmark. By contrast, the situation is different for $\hat h_{R}$, $\hat h_{A(\pi)}$ and $\hat h_{A(P)}$, shown in Figures~\ref{fig:sdp},  ~\ref{fig:sdpr} and~\ref{fig:ap}. Since these controllers are biased when the regularization parameters are finite, the UCRLB is no longer an informative benchmark. Instead, their performance should be compared with the lower bound $f(\theta,b,b')$, which accounts for bias and therefore applies to the average risk. 

Figure~\ref{fig:ap} shows that the average risk minimizing controller $\hat h_{A(P)}$ achieves the best overall performance among the controllers considered. Its risk follows a similar profile and remains close to the lower bound $f(\theta,b,b')$ throughout the parameter domain, and its average risk is the smallest among the four controllers. Moreover, $\hat h_{A(\pi)}$ and $\hat h_{A(P)}$ can be compared from a unified Bayesian perspective, where $\hat h_{A(\pi)}$ corresponds to a Gaussian prior, and $\hat h_{A(P)}$ corresponds to a uniform prior over the parameter domain. Since the Gaussian prior assigns more probability mass near its mean, $\hat h_{A(\pi)}$ emphasizes the center of the domain more than the boundary. Consequently, as shown in Figure~\ref{fig:ex3-all-risks} and Figures~\ref{fig:sdpr} and~\ref{fig:ap}, $\hat h_{A(\pi)}$ performs better than $\hat h_{A(P)}$ near the center, but worse near the boundary. This illustrates how the choice of prior affects performance over the parameter domain. 

Finally, these observations also illustrate the ``waterbed'' effect in data-driven LQR. In Figure~\ref{fig:sdpr}, the substantial reduction of risk around $A=0.5$ is compensated by a pronounced increase in risk for values of $A$ farther away from $0.5$. In contrast, Figure~\ref{fig:ap} shows a more moderate risk reduction near the center, accompanied by a smaller increase elsewhere. Thus, these observations show that reducing the risk in one region of the parameter domain is accompanied by increased risk in other regions. The proposed lower bound makes this redistribution effect explicit and provides a quantitative way to assess how severe it is.

	\vspace{-3mm}
	\section{Conclusions} \label{Sct8}

Motivated by fundamental limits in point estimation, this work developed a statistical decision framework for data-driven control. By combining the bias-variance decomposition with the CRLB, we derived fundamental lower bounds on the performance for data-driven control problems under quadratic loss. These results revealed quantitative limitations that no data-driven controller design can circumvent and make the bias--variance tradeoff explicit. They also reveal a ``waterbed'' effect in data-driven control, analogous to the fundamental limitation of sensitivity shaping in classical control.

The two case studies, namely the feedforward control problem and the benchmark LQR problem, illustrated several implications of the framework. First, because the derived lower bounds depend explicitly on the inverse of Fisher information matrix, they show that the statistical difficulty of data-driven control is fundamentally shaped by the underlying system itself and the information available in the data. Second, the results show that no controller can be uniformly best over the entire parameter domain. Accordingly, controller design should be guided by the region of the parameter space most relevant to the intended application.

In future, we will extend the proposed framework to analyze the statistical limitations of data-driven MPC, reinforcement learning, and online control methods.
	
%
%
	\section*{References}
	\bibliographystyle{IEEEtran}
	\bibliography{refs}

	\appendices 
	\numberwithin{equation}{section}
	\vspace{-3mm}

    \section{Cram\'er-Rao Lower Bound (Theorem~\ref{Thm:CRLB})} \label{AppA_CRLB}

	\begin{proof}
		Using the bias-variance decomposition of risk in \eqref{risk-decomposition}, we see that the key step is to show that
		\begin{equation}\label{eq:CRLB}
			\Var_\theta[\hat{h}(Z^N)]
			\succcurlyeq
			m'(\theta) I_{F,N}^{-1}(\theta) \bigl(m'(\theta)\bigr)^\top,
			\quad \forall \theta\in \mathcal{D}_\theta.
		\end{equation}
		
		We first show the result when the support of $p(Z^N; \theta)$ does not depend on $\theta$. The matrix version of Cauchy-Schwarz inequality (see \cite[Th. 1]{Tripathi1999matrix}) gives
\begin{align*}
	\Var_\theta\left[\hat{h}(Z^N)\right] \succcurlyeq & \Cov_\theta\left[\hat{h}(Z^N), q_\theta(Z^N)\right] \Var_\theta^{-1}\left[q_\theta(Z^N)\right] \times\\
&\Cov_\theta\!\left[q_\theta(Z^N), \hat{h}(\theta)(Z^N)\right]
\end{align*}
for any random vector $q_\theta(Z^N)$ with finite second moments. Taking $q_\theta(Z^N)=S_N(Z^N,\theta)$ gives that
$\Var_\theta[q_\theta(Z^N)] = I_{F,N}(\theta)$. Furthermore, under suitable regularity conditions, we have that 
\begin{align*}
	\Cov_\theta\left[\hat{h}(Z^N), S_N(Z^N,\theta)\right]&=\mathbb{E}_\theta\left[\hat{h}(Z^N)S_N(Z^N,\theta)\right] \\
	&= \frac{\partial}{\partial \theta^\top}\mathbb{E}_\theta\left[\hat{h}(Z^N)\right],
\end{align*}
where the first equality is due to that $S_N(Z^N;\theta)=\frac{\partial \log p(Z^N;\theta)}{\partial \theta}$ has zero mean, and the second equality comes from 
\begin{align*}
	&\mathbb{E}_\theta\left[h(Z^N)\frac{\partial}{\partial \theta^\top}\log p(Z^N;\theta)\right] \\
= &\int \hat{h}(Z^N)\lim_{\Delta\to 0}
\frac{p(z;\theta+\Delta)-p(z;\theta)}{\Delta}dz \\
= &\lim_{\Delta\to 0}
\frac{\mathbb{E}_{\theta+\Delta}[h(Z^N)]-\mathbb{E}_\theta[h(Z^N)]}{\Delta} \\
= &\frac{\partial}{\partial \theta^\top}\mathbb{E}_\theta\!\left[h(Z^N)\right].
\end{align*}
In this way, we have 
\begin{equation*}
	\begin{split}
		\Var_\theta\left[\hat{h}(Z^N)\right] &\succcurlyeq \frac{\partial}{\partial \theta^\top}\mathbb{E}_\theta[\hat{h}(Z^N)] I_{F,N}^{-1}(\theta) \bigl(\frac{\partial}{\partial \theta^\top}\mathbb{E}_\theta[\hat{h}(Z^N)]\bigr)^\top \\
		&=m'(\theta) I_{F,N}^{-1}(\theta) \bigl(m'(\theta)\bigr)^\top.
	\end{split}	
\end{equation*}

For the setting in \eqref{eq:Bayes_pdfs}, in which $\varepsilon_t(Z^t;\theta)$ depends explicitly on $\theta$, the regularity conditions must be adjusted accordingly. With these modifications in place, the CRLB in Theorem~\ref{Thm:CRLB} follows by an argument analogous to that used in the above proof. We omit the details for brevity.
	\end{proof}

    \vspace{-3mm}
	\section{Euler-Lagrange Equation (Theorem~\ref{Theorem-EL-Equation-vector})} \label{AppB_EL_Equation}

\begin{proof}
We first prove Theorem~\ref{Theorem-EL-Equation-vector} in the scalar case $n_\theta=n_h=1$, with parameter domain $\mathcal{D}_\theta=[\theta_1,\theta_2]$. To be specific, we show that under Assumptions~\ref{Assp1} and \ref{Assp2}, if a minimizer $b_\theta^*$ of the variational problem \eqref{Bayes-risk-CoV-scalar} exists, then it is unique and satisfies the Euler-Lagrange equation
\begin{equation}  \label{CoV-Scalar-EL-Equation}
	\frac{d}{d\theta}\left(\frac{\partial f}{\partial b_\theta'}\right) - \frac{\partial f}{\partial b_\theta} = 0, \qquad \forall \theta \in \left(\theta_1,\theta_2\right), 
\end{equation}
subject to the natural boundary conditions
\begin{equation} \label{CoV-Scalar-EL-Natural_BC}
	\left.\frac{\partial f}{\partial b_\theta'}\right|_{\theta = \theta_1}=0, \quad \left.\frac{\partial f}{\partial b_\theta'}\right|_{\theta = \theta_2}=0.
\end{equation}

We begin by calculating the variation $\tilde{F}$ of the functional \eqref{Bayes-risk-CoV-scalar}. Let $\tilde{b}_\theta$ and $\tilde{b}_\theta'$ be the variation and its derivative. Then, the increment $\tilde{F}= F(b_\theta+\tilde{b}_\theta)-F(b_\theta)$ is
\begin{equation*}
\begin{split}
   \tilde{F} &= \int_{\theta_1}^{\theta_2}\left(f\bigl(\theta, b_\theta+\tilde{b}_\theta, b_\theta'+\tilde{b}_\theta'\bigr)-f\bigl(\theta,b_\theta,b_\theta'\bigr)\right) d\theta \\
     & =\int_{\theta_1}^{\theta_2}\left(\frac{\partial f}{\partial b_\theta}\tilde{b}_\theta + \frac{\partial f}{\partial {b}_\theta'}\tilde{b}_\theta'\right)d\theta + \cdots,
\end{split}
\end{equation*}
where the second equality is due to the Taylor's expansion, and the dots denote terms of order higher than 1 w.r.t. $\tilde{b}_\theta$ and $\tilde{b}_\theta'$. After omitting high order terms, the principal linear part of the increment $\tilde{J}$ is
\begin{equation*}
  \tilde{F} = \int_{\theta_1}^{\theta_2}\left(\frac{\partial f}{\partial b_\theta}\tilde{b}_\theta + \frac{\partial f}{\partial {b}_\theta'}\tilde{b}_\theta'\right) d\theta.
\end{equation*}
Integration by parts now gives
\begin{equation*}
  \begin{split}
    \tilde{F} = &\int_{\theta_1}^{\theta_2}\left(\frac{\partial f}{\partial b_\theta} - \frac{d}{d\theta}\left(\frac{\partial f}{\partial {b}_\theta'}\right)\right) \tilde{b}_\theta  d\theta + \frac{\partial f}{\partial {b}_\theta'} \tilde{b}_\theta\Big|_{\theta = \theta_1}^{\theta = \theta_2} \\
    = &\int_{\theta_1}^{\theta_2}\left(\frac{\partial f}{\partial b_\theta} - \frac{d}{d\theta}\left(\frac{\partial f}{\partial {b}_\theta'}\right)\right) \tilde{b}_\theta d\theta  + \\
	&\frac{\partial f}{\partial {b}_\theta'}\bigl|_{\theta = \theta_2}\tilde{b}_\theta(\theta_2) - \frac{\partial f}{\partial {b}_\theta'}\bigl|_{\theta = \theta_1}\tilde{b}_\theta(\theta_1).
  \end{split}
\end{equation*}
Since $b$ should be an extremal, we have $\tilde{f}$ = 0, requiring that
\begin{subequations}
	\begin{align}
		\frac{\partial f}{\partial b_\theta} - \frac{d}{d\theta}\left(\frac{\partial f}{\partial {b}_\theta'}\right) &= 0, \label{eq:EL-equation}\\
		\frac{\partial f}{\partial {b}_\theta'}\bigl|_{\theta = \theta_2}\tilde{b}_\theta(\theta_2) &= 0, \quad \frac{\partial f}{\partial {b}_\theta'}\bigl|_{\theta = \theta_1}\tilde{b}_\theta(\theta_1) = 0, \label{eq:NB}
	\end{align}
\end{subequations}
where the first condition \eqref{eq:EL-equation} is the Euler-Lagrange equation. Moreover, in contrast to the fixed end-point problem, $b_\theta$ is not required to vanish at $\theta_1$ and $\theta_2$, and $\tilde{b}_\theta$ is arbitrary. Consequently, the natural boundary condition \eqref{CoV-Scalar-EL-Natural_BC} is necessary to guarantee that $\tilde{F}=0$. Furthermore, since $W(\theta) \succ 0$ and $I_F(\theta) \succ 0$ for $\theta \in \mathcal{D}_\theta$, we have that the function $f(\theta,b_\theta,b_\theta')$ is strictly convex in $(b_\theta,b_\theta')$ for every $\theta \in \mathcal{D}_\theta$. Based on \cite[Th. 2.1]{Dacorogna2024introduction}, we then conclude that any sufficiently smooth admissible function satisfying
\eqref{eq: general form of EL equation-vector} and \eqref{eq:natural boundary condition-vector} is the unique minimizer of the variational problem \eqref{Bayes-risk-CoV-scalar}.

The extension to the multivariate setting is analogous and is omitted. 
\end{proof}

    \vspace{-3mm}
	\section{Fisher Information Matrix for LQR} \label{AppC_Fisher}

The likelihood of the state sequence $p(Z^N;\theta)$ can be decomposed using the Markovian property:
\begin{equation}
   p(Z^N;\theta)= p(x_0) \prod_{k=0}^{N-1} p(x_{k+1} | x_k, u_k; \theta).
\end{equation}
The log-likelihood is
\begin{equation*}
   L_N(Z^N;\theta) = \log p(x_0) + \sum_{k=0}^{N-1} \log  p(x_{k+1} | x_k, u_k; \theta).
\end{equation*}
Since the noise $e_k \sim \mathcal{N}(0, \sigma_e^2)$, the conditional distribution of $x_{k+1}$ is $(x_{k+1} | x_k, u_k; \theta) \sim \mathcal{N}(A x_k + B u_k, \sigma_e^2)$. The conditional log-likelihood term is
\begin{align*}
   \log p(x_{k+1} | x_k, u_k; \theta) = &-\frac{1}{2\sigma_e^2} \left(x_{k+1} - A x_k - B u_k\right)^2 - \\
   & \frac{1}{2}\log(2\pi\sigma_e^2) .
\end{align*}
In this way, $L_N(Z^N; \theta)$ can be rewritten as
\begin{align*}
   L_N(Z^N; \theta) = &-\sum_{k=0}^{N-1} \frac{1}{2\sigma_e^2} \left(x_{k+1} - A x_k - B u_k\right)^2 + 
   \\
   &\log p(x_0) - \frac{N}{2}\log(2\pi\sigma_e^2) .
\end{align*}
Then, the score function $S_N(Z^N; \theta)$ is
\begin{equation}
	\begin{split}
		S_N(Z^N; \theta)  & = \frac{\partial}{\partial \theta}L_N(Z^N; \theta)	 \\
		& = \sum_{k=0}^{N-1} \frac{1}{\sigma_e^2} \left(x_{k+1} - A x_k - B u_k\right) \begin{bmatrix} x_k \\ u_k \end{bmatrix} \\
        &= \sum_{k=0}^{N-1} \frac{e_k}{\sigma_e^2} \begin{bmatrix} x_k \\ u_k \end{bmatrix}.
	\end{split}   
\end{equation}
Accoring to \eqref{eq:FIM}, the Fisher information matrix is equal to
\begin{equation*}
	\begin{split}
		I_{F,N}(\theta) &= \mathbb{E}_\theta\left[S_N(Z^N; \theta) S_N^\top(Z^N; \theta)\right] \\
		& = \mathbb{E}_\theta\left[ \left( \sum_{k=0}^{N-1} \frac{e_k}{\sigma_e^2} \begin{bmatrix} x_k \\ u_k \end{bmatrix} \right) \left( \sum_{j=0}^{N-1} \frac{e_j}{\sigma_e^2} \begin{bmatrix} x_j \\ u_j \end{bmatrix}^\top\right) \right] \\
		&= \sum_{k=0}^{N-1} \frac{1}{\sigma_e^4} \mathbb{E}_\theta[e_k^2] \mathbb{E}_\theta\left[ \begin{bmatrix} x_k^2 & x_k u_k \\ u_k x_k & u_k^2 \end{bmatrix} \right] \\
		&= \frac{1}{\sigma_e^2} \sum_{k=0}^{N-1} \begin{bmatrix} \mathbb{E}_\theta[x_k^2] & 0 \\ 0 & \mathbb{E}_\theta[u_k^2] \end{bmatrix},
	\end{split}   
\end{equation*}
where the third and fourth equalities follow from the fact that $\{e_k\}$ is i.i.d. and independent of $\{x_j,u_j\}_{j\le k}$, and from the fact that $x_k$ and $u_k$ are zero-mean and mutually independent. Furthermore, with $\mathbb{E}_\theta[x_0] = 0$ and $\mathbb{E}_\theta[u_k] = 0$, the mean of the state $\mathbb{E}_\theta[x_k]$ is zero for all $k$. The state variance, denoted by $\sigma_k^2 = \mathbb{E}_\theta[x_k^2]$, follows the recurrence relation:
\begin{equation}
   \mathbb{E}_\theta[x_{k+1}^2] = \mathbb{E}_\theta[(A x_k + B u_k + e_k)^2].
\end{equation}
Due to the independence and zero-mean nature of $x_k$, $u_k$, and $e_k$, the cross-terms vanish. We therefore have
\begin{equation}
	\begin{split}
		\sigma_{k+1}^2 &= A^2 \mathbb{E}_\theta[x_k^2] + B^2 \mathbb{E}_\theta[u_k^2] + \mathbb{E}_\theta[e_k^2] \\
		&= A^2 \sigma_k^2 + B^2 \sigma_u^2 + \sigma_e^2.
	\end{split}
\end{equation}
The solution to this recursion from $k=0$ to $k=N-1$ is found by iterating:
\begin{equation}
\sigma_k^2 = A^{2k} \sigma_0^2 + \sum_{j=0}^{k-1} A^{2(k-1-j)} (B^2 \sigma_u^2 + \sigma_e^2).
\end{equation}
The components of the Fisher information matrix are
\begin{align}
   \frac{1}{\sigma_e^2} \sum_{k=0}^{N-1} \mathbb{E}_\theta[u_k^2] &= \frac{1}{\sigma_e^2} \sum_{k=0}^{N-1} \sigma_u^2 = \frac{N \sigma_u^2}{\sigma_e^2}, \\
   \frac{1}{\sigma_e^2} \sum_{k=0}^{N-1} \mathbb{E}_\theta[x_k^2] &= \frac{1}{\sigma_e^2} \sum_{k=0}^{N-1} \sigma_k^2.
\end{align}
The general, $I_{F,N}(\theta)$, by evaluating the geometric sum, is
$\frac{1}{\sigma_e^2}\text{diag}\left(\sum_{k=0}^{N-1} ( A^{2k} \sigma_0^2 + \frac{1 - A^{2k}}{1 - A^2} (B^2 \sigma_u^2 + \sigma_e^2)),N \sigma_u^2 \right)$.

\vspace{-3mm}
\section{Solving Calculus of Variational for LQR} \label{AppD_CoV_LQR}

For the optimal-bias lower bound of the scalar case, the cost functional is given by
\begin{equation} \label{E14}
   F\left(b_\theta\right) = \int_{\mathcal{D}_{\theta}} f\left(\theta,b_\theta,b_\theta'\right),
\end{equation}
where
\begin{align*}	
    f(\theta,b_\theta,b_\theta')  &= b_\theta^2W(\theta) + W(\theta)(m'(\theta))^\top I_{F,N}^{-1}(\theta)m'(\theta),\\
    m'(\theta) &= b_\theta'+ \frac{\partial K}{\partial \theta}, W(\theta) = \Sigma_{K}\otimes(R+B^2P), \\
\end{align*}
with $\Sigma_{K} = \frac{\sigma_e^2}{1-(A-BK)^2}$, $b_\theta' = \begin{bmatrix}
   \frac{\partial b_\theta}{\partial A} \\ \frac{\partial b_\theta}{\partial B}\end{bmatrix} := \begin{bmatrix}b_A \\ b_B\end{bmatrix}$, $\frac{\partial K}{\partial \theta} = \begin{bmatrix} \frac{\partial K}{\partial A} \\ \frac{\partial K}{\partial B}\end{bmatrix}$,
\begin{align*}  
	\frac{\partial K}{\partial A} &= - \frac{BPR+B^3P^2 + ABR \frac{\partial P}{\partial A}}{(R+B^2P)^2}, \\
    \frac{\partial K}{\partial B} &= - \frac{AP(R+B^2P) + ABR \frac{\partial P}{\partial B} - 2AB^2P^2}{(R+B^2P)^2}, \\
	\frac{\partial P}{\partial A} &= \frac{\frac{2APR}{R+B^2P}}{1 - A^2 + \frac{A^2B^2 (2PR + B^2P^2)}{(R+B^2P)^2}}, \\
    \frac{\partial P}{\partial B} &= - \frac{\frac{2A^2 B P^2 R}{(R+B^2P)^2}}{1 - A^2 + \frac{A^2B^2 (2PR + B^2P^2)}{(R+B^2P)^2}},
\end{align*}
where derivatives $\frac{\partial K}{\partial A}$, $\frac{\partial K}{\partial B}$, $\frac{\partial P}{\partial A}$ and $\frac{\partial P}{\partial A}$ come from \eqref{eq:LQR_Riccati} and \eqref{eq:LQR_feedback_gain}. The Euler-Lagrange equation \eqref{eq: general form of EL equation-vector} in becomes
\begin{align}
    \frac{\partial f}{\partial b} - \frac{\partial}{\partial A}\frac{\partial f}{\partial b_A} - \frac{\partial}{\partial B}\frac{\partial f}{\partial b_B} =0.
\end{align}
Moreover, the natural boundary condition \eqref{eq:natural boundary condition-vector} becomes
\begin{align}
\frac{\partial f}{\partial b_A}n_a+\frac{\partial f}{\partial b_B}n_b=0, \ \text{on} \ \partial \mathcal{D}_\theta,
\end{align}
where $n=(n_a,n_b)$ is the outward unit normal.
	
\end{document}